\documentclass[journal]{IEEEtran}
\usepackage{cite}

\usepackage{graphicx}
\usepackage{booktabs}
\usepackage{multirow}
\usepackage{color}

\usepackage{amsmath}
\usepackage{amsthm}
\newtheorem{theorem}{\textit{Theorem}}

\newtheorem{assumption}{\textit{Assumption}}

\usepackage{algorithm}
\usepackage{algorithmic}

\usepackage{array}

\ifCLASSOPTIONcompsoc
 \usepackage[caption=false,font=normalsize,labelfont=sf,textfont=sf]{subfig}
\else
 \usepackage[caption=false,font=footnotesize]{subfig}
\fi
\usepackage{fixltx2e}
\usepackage{dblfloatfix}

\ifCLASSOPTIONcaptionsoff
 \usepackage[nomarkers]{endfloat}
\let\MYoriglatexcaption\caption
\renewcommand{\caption}[2][\relax]{\MYoriglatexcaption[#2]{#2}}
\fi
\usepackage{url}

\begin{document}
%
\title{Co-optimizing Bidding and Power Allocation of an EV Aggregator Providing Real-time Frequency Regulation Service}
%
%
%

\author{Ruike Lyu,~\IEEEmembership{Student Member,~IEEE}, Hongye Guo,~\IEEEmembership{Member,~IEEE}, Kedi Zheng,~\IEEEmembership{Member,~IEEE}, Mingyang Sun,~\IEEEmembership{Member,~IEEE}, and Qixin Chen,~\IEEEmembership{Senior Member,~IEEE}
\thanks{This work was supported in part by the International (NSFC-NWO) Joint Research Project of National Natural Science Foundation of China under Grant 52161135201, and in part by National Natural Science Foundation of China under Grant 52107102.

R. Lyu, H. Guo, K. Zheng, and Q. Chen are with the State Key Lab of Power Systems, the Department of Electrical Engineering, Tsinghua University, Beijing 100084, China. M. Sun is with the State Key Laboratory of Industrial Control Technology, Department of Control Science and Engineering, Zhejiang University, Hangzhou 310027, China.}}

\maketitle

\begin{abstract}
  The rapidly expanding scale of electric vehicle (EV) fleets and continuously decreasing battery costs are making vehicle-to-grid services a reality. In this paper, we study the interaction between the problems of an EV aggregator's bidding in the regulation market and power allocation (i.e., determining the (dis)charging powers of the EVs in regulation deployment). Although the two problems are coupled, they are often regarded as decoupled and optimized separately for complexity issues. However, failing to consider the coupling of bidding and power allocation can lead to a decline in the profit of the EV aggregator (EVA). In this paper, we propose a framework for co-optimizing EVA bidding and power allocation in the regulation market. The bidding model is formulated as a stochastic programming problem with embedded power allocation in discretized regulation signal scenarios. To meet the solution time requirement for regulation deployment, we further propose a power allocation model that can be solved online. It utilizes the Lagrange multipliers from the bidding problem to ensure that the allocation results correspond to the optimal solution of the bidding problem. The effect of the proposed framework on improving EVA profits and reducing degradation costs is verified in the case study.
\end{abstract}


\ifCLASSOPTIONpeerreview
\begin{center} \bfseries EDICS Category: 3-BBND \end{center}
\fi
%
\IEEEpeerreviewmaketitle
\section*{Nomenclature}
\addcontentsline{toc}{section}{Nomenclature}

\subsection*{Sets and Indices}
\begin{IEEEdescription}[\IEEEusemathlabelsep\IEEEsetlabelwidth{$k/K/\mathcal{K}$}]
  \item[$i/I$] Index/set of EVs.
  \item[$k/K$] Index/set of power segments.
  \item[$s/S$] Index/set of price scenarios.
  \item[$\hat{s}/\hat{S}$] Index/set of regulation signal scenarios.
  \item[$t/T$] Index/set of time intervals.
  \item[$\hat{t}$] Index of time within time interval $t$.
  \item[${\rm cur}$] Superscript for current time.
  \item[${\rm ch}$] Superscript for charging.
  \item[${\rm dis}$] Superscript for discharging.
  \item[${\rm deg}$] Superscript for degradation.
  \item[${\rm deploy}$] Superscript for deployment.
  \item[${\rm e}$] Superscript for energy.
  \item[${\rm cap}$] Superscript for capacity.
  \item[${\rm mil}$] Superscript for mileage.
\end{IEEEdescription}

\subsection*{Parameters and Constants}
\begin{IEEEdescription}[\IEEEusemathlabelsep\IEEEsetlabelwidth{$k/K/\mathcal{K}$}]
  \item[$\Delta t$] Length of time interval (hour, h).
  \item[$\Delta \hat{t}$] Regulation signal interval (second, s).
  \item[$\Delta t^{\rm req}$] Required time at the fully deployed output level (h).
  \item[$\Delta t^{\rm cur}$] Remaining time of current time interval (h).
  \item[$\Delta t^{\rm plug}_i$] Remaining plug-in time of EV $i$ before leaving (h).
  \item[$\delta_{\hat{s}}$] Value of regulation signal in scenario $\hat{s}$.
  \item[$\eta^{\rm dis(ch)}_i$] Discharging (charging) efficiency of EV $i$.
  \item[$\pi_{t, \hat{s}}$] Probability of regulation signal scenario $\hat{s}$ at time interval $t$.
  \item[$\pi_s$] Probability of price scenario $s$.
  \item[$a^{\rm mil}_t$] Resource-specific mileage multiplier for regulation at time interval $t$.
  \item[$E^{\rm cur}_i$] Current state of charge (SOC) of EV $i$ (kWh).
  \item[$E^{\rm min(max)}_{t, i}$] Lower (upper) SOC limit of EV $i$ at time interval $t$ (kWh).
  \item[$E^{\rm leave}_i$] Required SOC of EV $i$ when leaving (kWh).
  \item[$M$] A large number.
  \item[$\overline{P^{\rm dis(ch)}_{i, k}}$] Upper power limit for discharging (charging) at power segment $k$ of EV $i$ (kW).
  \item[$Pr^{\rm deg}_{i, k}$] Degradation cost rate at power segment $k$ of EV $i$ (\$/kWh).
  \item[$Pr^{\rm e}_{s, t}$] Energy price at time interval $t$ in scenario $s$ (\$/MWh).
  \item[$Pr^{\rm cap}_{s, t}$] Regulation capacity price at time interval $t$ in scenario $s$ (\$/MWh).
  \item[$Pr^{\rm mil}_{s, t}$] Regulation mileage price at time interval $t$ in scenario $s$ (\$/MW).
  \item[$s^{\rm perf}$] Performance score of providing regulation mileage.
  \item[$t^{\rm cur}$] Current time interval.
  \item[$\hat{t}^{\rm cur}$] Current time within the current time interval.
  \item[$t^{\rm end}$] The ending time interval under consideration.
  \item[$u_{t, i}$] Plug-in state of EV $i$ at time interval $t$.

\end{IEEEdescription}

\subsection*{Variables and Lagrange Multipliers}
\begin{IEEEdescription}[\IEEEusemathlabelsep\IEEEsetlabelwidth{$Income^{\rm mil}_{s, t}$}]
  \item[$Cost^{\rm deg}_{t, \hat{s}}$] Degradation cost at time interval $t$ in regulation signal scenario $\hat{s}$ (\$).
  \item[$Cost_{\hat{t}, \hat{s}}$] Cost of providing regulation at time $\hat{t}$ in regulation signal scenario $\hat{s}$ (\$).
  \item[$E_{t, i}$] Expected SOC of EV $i$ at the beginning of time interval $t$ (kWh).
  \item[$Income^{\rm e}_{s, t}$] Energy income from the energy market at time interval $t$ in price scenario $s$ (\$).
  \item[$Income^{\rm ch}_{s, t}$] Income from the EV owners for charging (\$).
  \item[$Income^{\rm cap}_{s, t}$] Regulation capacity income at time interval $t$ in price scenario $s$ (\$).
  \item[$Income^{\rm mil}_{s, t}$] Regulation mileage income at time interval $t$ in price scenario $s$ (\$).
  \item[$Income^{\rm deploy}_{t, s, \hat{s}}$] Energy income for deploying regulation at time interval $t$ in price scenario $s$ with regulation signal scenario $\hat{s}$ (\$).
  \item[$Profit(t)$] Total profit for future time intervals starting from time intervals $t$ (\$).
  \item[$Profit(\hat{t})$] Total profit for future time intervals starting from time $\hat{t}$ within time intervals $t$ (\$).
  \item[$P^{\rm dis(ch)}_{t, \hat{s}, i, k}$] Discharging (charging) power at segment $k$ of EV $i$ at time interval $t$ in regulation scenario $\hat{s}$ (kW).
  \item[$P^{\rm dis(ch)}_{\hat{t}, \hat{s}, i, k}$] Discharging (charging) power at segment $k$ of EV $i$ at time $\hat{t}$ in regulation scenario $\hat{s}$ (kW).
  \item[$P_{t}$] Energy bid at time interval $t$ (MW).
  \item[$R_{t}$] Regulation capacity bid at time interval $t$ (MW).
  \item[$\Delta P^{\rm dis(ch)}_{\hat{t}, \hat{s}}$] The power deviated from the regulation signal in the direction of discharging (charging) at time $\hat{t}$ in regulation signal scenario $\hat{s}$ (MW).
  \item[$\lambda^{\rm E}_{t, i}$] Lagrange multiplier of the SOC constraint regarding EV $i$ at time interval $t$ (\$/kWh).
  \item[$\lambda^{\rm bal}_{t, \hat{s}}$] Lagrange multiplier of the power balance constraint regarding regulation signal scenario $\hat{s}$ at time interval $t$ (\$/kW).
  \item[$\lambda^{\rm bal}_{\hat{t}, \hat{s}}$] Lagrange multiplier of the power balance constraint regarding regulation signal scenario $\hat{s}$ at time $\hat{t}$ (\$/kW).
  \item[$\mu^{\rm dis(ch)}_{t, \hat{s}, i, k}$] Lagrange multiplier of the discharging (charging) power limit constraint regarding segment $k$ of EV $i$ in regulation signal scenario $\hat{s}$ at time interval $t$ (\$/kW).
  \item[$\mu^{\rm dis(ch)}_{\hat{t}, \hat{s}, i, k}$] Lagrange multiplier of the discharging (charging) power limit constraint regarding segment $k$ of EV $i$ in regulation signal scenario $\hat{s}$ at time $\hat{t}$ (\$/kW).

\end{IEEEdescription}

\section{Introduction}
%
%
%
%
\IEEEPARstart{V}{ehicle-to-grid} (V2G) services are becoming a reality. 
In the context of the decarbonization of the power system and the electrification of the transportation system, the scale of electric vehicle (EV) fleets is rapidly expanding. 
By 2030, the number of EVs in China is expected to reach 80 million~\cite{zheng_systematic_2020}. Additionally, the volume-weighted pack and cell price of EVs fell by more than 75\% from 2014 to 2021~\cite{bloomberg2021electric}. 
Considerable battery capacity and continuously reduced battery degradation costs will make it economical for EVs to support the operation of the smart grid in the near future~\cite{zheng_economic_2022}, while the electricity market will provide economic incentives for V2G implementations~\cite{teng_technical_2020}. 
As a result of drivers' habits, especially the habits of the current mainstream nighttime charging~\cite{mckinsey_report_2018}, the plug-in time of EVs is usually much longer than the time required for charging. 
The resulting charging flexibility, i.e., the time of charging or discharging can be adjusted, enables EVs to provide load shift~\cite{gan_optimal_2013}, frequency regulation~\cite{kaur_coordinated_2019}, reserve~\cite{duan_bidding_2021} and other services for the grid. 
However, supporting the operation of the grid may bring certain losses to EV owners, such as battery degradation\cite{sun_real-time_2014}, insufficient battery state of charge (SOC) for driving demand~\cite{peng_optimal_2017}, etc., which can be compensated by the income from participating in the electricity market.

Due to the limited capacity of a single EV, EV aggregators (EVAs) usually need to aggregate a large number of EVs to meet the threshold for electricity market participation~\cite{ansari_coordinated_2015}. 
Such an EVA-led charging management mode has received widespread attention and is roughly divided into three links: the EVA bidding in the market\cite{sortomme_optimal_2011,sortomme_optimal_2012, jin_optimizing_2013, vagropoulos_optimal_2013, ansari_coordinated_2015, sarker_optimal_2016, vatandoust_risk-averse_2019, duan_bidding_2021}, the EVA allocating (dis)charging power to EVs\cite{escudero-garzas_fair_2012, sun_real-time_2014, peng_optimal_2017, vagropoulos_real-time_2016}, and the revenue distribution between the EVA and EV owners~\cite{li2018robust, chen2019non}. 
The main idea of EVA bidding is to formulate a strategy to maximize the profit subject to the EV charging constraints and driving demand of EV owners.
The bidding model for the unidirectional and bidirectional charging modes were formulated in~\cite{sortomme_optimal_2011} and~\cite{sortomme_optimal_2012}, respectively.
The bidding for an EVA with energy storages was studied in~\cite{jin_optimizing_2013}.
To cope with the uncertainty of market prices and regulation signals, the EVA optimal bidding strategy can be formulated using stochastic optimization~\cite{vagropoulos_optimal_2013}, fuzzy optimization~\cite{ansari_coordinated_2015}, robust optimization~\cite{li2018robust}, and the (conditional) value-at-risk method~\cite{vatandoust_risk-averse_2019, duan_bidding_2021}. 

\begin{table*}[!t]
  \renewcommand{\arraystretch}{1.0}
  \caption{Comparison of Relevant Literature}
  \label{tab_reference}
  \centering
  \setlength{\tabcolsep}{1mm}{
    {
      \begin{tabular}{cccc}
        \toprule
        \begin{tabular}[c]{@{}c@{}}Ref. on EVA  Bidding\end{tabular}   &
        \begin{tabular}[c]{@{}c@{}}Consideration for  Power Allocation\end{tabular}   &
        \begin{tabular}[c]{@{}c@{}}Ref. on Power  Allocation\end{tabular}   &
        \begin{tabular}[c]{@{}c@{}}Consideration for Bidding Profit\end{tabular}   \\   \midrule
        \begin{tabular}[c]{@{}c@{}}~\cite{ansari_coordinated_2015,sortomme_optimal_2011, sortomme_optimal_2012, vagropoulos_optimal_2013,   jin_optimizing_2013,vatandoust_risk-averse_2019, duan_bidding_2021}\end{tabular}  &
        \begin{tabular}[c]{@{}c@{}}Allocated Proportional  to the\\ allocated regulation capacity\end{tabular}  &
        \begin{tabular}[c]{@{}c@{}}~\cite{escudero-garzas_fair_2012, sun_real-time_2014, vagropoulos_real-time_2016} \end{tabular}  &
        \begin{tabular}[c]{@{}c@{}}Not considered  or not related\end{tabular}    \\
        ~\cite{sarker_optimal_2016} &
        \begin{tabular}[c]{@{}c@{}}Evenly allocated to the EVs\end{tabular}  &
        ~\cite{peng_optimal_2017}   &
        \begin{tabular}[c]{@{}c@{}}Maximizing the regulation \\capacity at the next time interval\end{tabular}    \\
        \begin{tabular}[c]{@{}c@{}}Our bidding model\end{tabular}  &
        \begin{tabular}[c]{@{}c@{}}Optimized simultaneously \end{tabular}  &
        \begin{tabular}[c]{@{}c@{}}Our power allocation\end{tabular}  &
        \begin{tabular}[c]{@{}c@{}}Optimal for the bidding profit\end{tabular}    \\ \bottomrule
      \end{tabular}
    }}
\end{table*}

After bidding, the EVA needs to allocate the power for deploying regulation, i.e., to determine the (dis)charging power of the EVs so that its net output can accurately follow the regulation signals issued by the independent system operators (ISO). 
Conceptually, in the EVA bidding problem, it is also necessary to consider the power allocation strategy because different allocation strategies will induce different costs, thus leading to differing EVA profits.
However, due to the uncertainty of regulation signals, it remains a challenge to reasonably model the impact of power allocation in EVA bidding, especially when faced with a large number of EVs. As a compromise, in the above studies on EVA bidding, the power allocation strategies are usually simplified to make the model complexity and computation time acceptable. 
For example, in ~\cite{sortomme_optimal_2012}, the aggregated regulation capacity is divided and allocated to each EV, and the change in EV (dis)charging power in response to regulation signals is set proportional to the regulation capacity allocated to the EVs. 
This is essentially a heuristic power allocation strategy, which is likely to deviate from the optimal allocation and therefore lead to a decline in the EVA profit. 
In \cite{li_emission-concerned_2013}, to reduce the model size, the aggregated EV fleet constraints replace the constraints of each EV, which may make the actual power allocation infeasible. 

Under the pay-for-performance mechanism\cite[Section 3.2.1]{bjm2022bpm11}, following the regulation signals that are issued every few seconds means that the (dis)charging power of thousands of EVs needs to be determined online.
A limited number of studies focus on allocating power for deploying regulation between EVs, and the main idea is to optimize the deployment cost while meeting the driving demand of EV owners.
In~\cite{escudero-garzas_fair_2012} and~\cite{sun_real-time_2014}, the power for deploying regulation is allocated with the goals of the fairness to EV owner income and of maximizing the welfare of EV owners, respectively, which are inconsistent with the goal of maximizing EVA profit.
In~\cite{vagropoulos_real-time_2016}, charging weights for the EVs are set according to the remaining time of charging and the amount of electricity to be charged, which determine the charging priority of the EVs in power allocation.
In~\cite{peng_optimal_2017}, the aim of power allocation is to maximize the net capacity for providing regulation at the next time interval. 
However, the impact of power allocation on the EVA profit is not considered or considered only heuristically in these studies, which also leads to a decline in the EVA profit. 

In summary, there are gaps in existing research regarding the coupling between EVA bidding and power allocation (see Table~\ref{tab_reference}). 
 In the existing studies on EVA bidding, the naive strategies of proportionally or evenly allocating the power fail to utilize the heterogeneous characteristics of EVs to optimize the cost of providing regulation.
We argue that EVA bidding and power allocation are essentially coupled and should be co-optimized. The main reasons are as follows:
\paragraph{} As a result of the market operation process, bidding and power allocation are alternating and process coupled (Fig.~\ref{fig_framework_time}). The bids for a certain time interval determine the regulation power to be allocated in the interval, and the power allocation results determine the feasible bidding range for the next time intervals by affecting the SOCs of the EVs. Considering this coupling is more important in the real-time market environment, where the EVA can modify the bids based on the latest information, including forecasts for prices and regulation signals, as well as SOCs of EVs to increase its profits~\cite{ortegavazquez_optimal_2014}.
\paragraph{} Bidding and power allocation are to optimize EVA profits in the long term (within the future hours) and in the short term (within the next few seconds), respectively, and the profits in the different time horizons need to be balanced. Conceptually, any approach that optimizes short-term profits without considering long-term impacts is inherently a greedy algorithm that naturally may deviate from the optimal.
\paragraph{} In practice, failing to consider such coupling will lead to a non-negligible decline in the profits of the EVA, which is verified by the case study using real data.

To the best of our knowledge, there is no research that reasonably models this coupling and provides a practical framework that meets the time constraints required for regulation deployment. Here we propose a framework for co-optimizing the bidding and power allocation of an EVA in the regulation market. Power allocation is embedded in the bidding model, which is formulated as stochastic programming and can be solved by commercial solvers. 
Stochastic optimizations may become inappropriate for online applications, especially when the number of scenarios and time intervals is large. 
To meet the computation time requirement for the regulation deployment, we design a power allocation model that can be solved online. 
Using the Lagrange multipliers of the bidding problem as the problem parameter, the power allocation problem equivalently considers the impact of the allocation results on the expected profit of EVA bidding in different scenarios in the future.
We theoretically prove that for any given regulation signal, the optimal solution to the power allocation problem corresponds to the optimal solution to the bidding problem. 
In the case study, the effect of co-optimizing bidding and power allocation on increasing the EVA profit and reducing the battery degradation costs is verified.

Our contributions are as follows:

\begin{itemize}
  \item We propose a framework to co-optimize the bidding and power allocation of an EVA in the regulation market. The bidding model is formulated as stochastic optimization with embedded power allocation in discretized regulation signal scenarios.

  \item To meet the computation time requirement for regulation deployment, we propose a power allocation model using the Lagrange multipliers from the bidding model, which can be calculated online. This determines the allocated power that corresponds to the optimal solution of the bidding model.

  \item The case study verifies the effect of considering the coupling between bidding and power allocation on improving the EVA profit and reducing EV battery degradation costs, thus proving the necessity of reasonably considering this coupling.

\end{itemize}

The remainder of this paper is organized as follows: Section~\ref{sec_framework} describes the market environment and an overview of the proposed framework. Section~\ref{sec_model_bid} describes the optimal bidding model with embedded power allocation. Section~\ref{sec_model_allocation} describes the power allocation model based on Lagrange multipliers. Section~\ref{sec_case_study} describes the case study. Section~\ref{sec_conclusion} contains the conclusions and prospects for future research.

\section{Framework for EVA Participation in Real-time Regulation Market}\label{sec_framework}

\subsection{Market Framework}
Regulation, also known as secondary frequency control, is procured by the ISO to maintain a real-time balance between the generation and the load. We adopt a general real-time regulation market mechanism and do not distinguish between upward and downward regulation. Resources provide regulation by following regulation signals (also called auto generation control, AGC), which are issued by the ISO~\cite{bjm2022bpm11}, typically every 4 or 2 seconds. As compensation, the resources receive capacity payments and mileage (performance) payments.
The EVA is treated as a price taker that only bids for capacity in the real-time market~\cite{vatandoust_risk-averse_2019}. Specifically, the EVA can declare a regulation capacity with a minimum price, so that its accepted capacity is guaranteed to be the declared capacity.

\subsection{Compensation for EV owners}
We assume that the EVA signs contracts with the EV owners, providing for the following terms: a. The EV owners park and plug the EVs in the charging piles controlled by the EVA for a certain period; b. On the premise of meeting the charging demand of EV owners, the EVA participates in the electricity market by managing (dis)charging of EVs; c. The EV owners are compensated according to the actual (dis)charging of the EVs and the resulting battery degradation cost~\cite{sarker_optimal_2016}. 
The specific distribution of profits involves the cooperative game between the EVA and EV owners~\cite{chen_competition_2022}, which is beyond the scope of this paper.

Since the time of arrival and departure and the corresponding SOCs of the EVs can be provided by the above contracts, we regard these as deterministic parameters, and they are known by the EVA. In addition, with the charging demand given, the charging fee (possibly with a discount~\cite{vatandoust_risk-averse_2019}) can also be considered a constant, so it does not appear in our model, although it affects the EVA's profits.  Meanwhile, we assume that the equipment for (dis)charging control, metering, etc., meets the needs of the proposed framework, and we do not focus on these specific implementations.

\subsection{Overview of the Proposed Framework for Co-optimizing Bidding and Power Allocation}\label{sec_framework_overview}

Here, we present an overview of the proposed framework (Fig. \ref{fig_overview}), including the optimal bidding module (B) with an embedded power allocation submodule (A1) and the Lagrange-multiplier-based power allocation module (A2).

\begin{figure}[!t]
  \centering
  \includegraphics[width=3.5in]{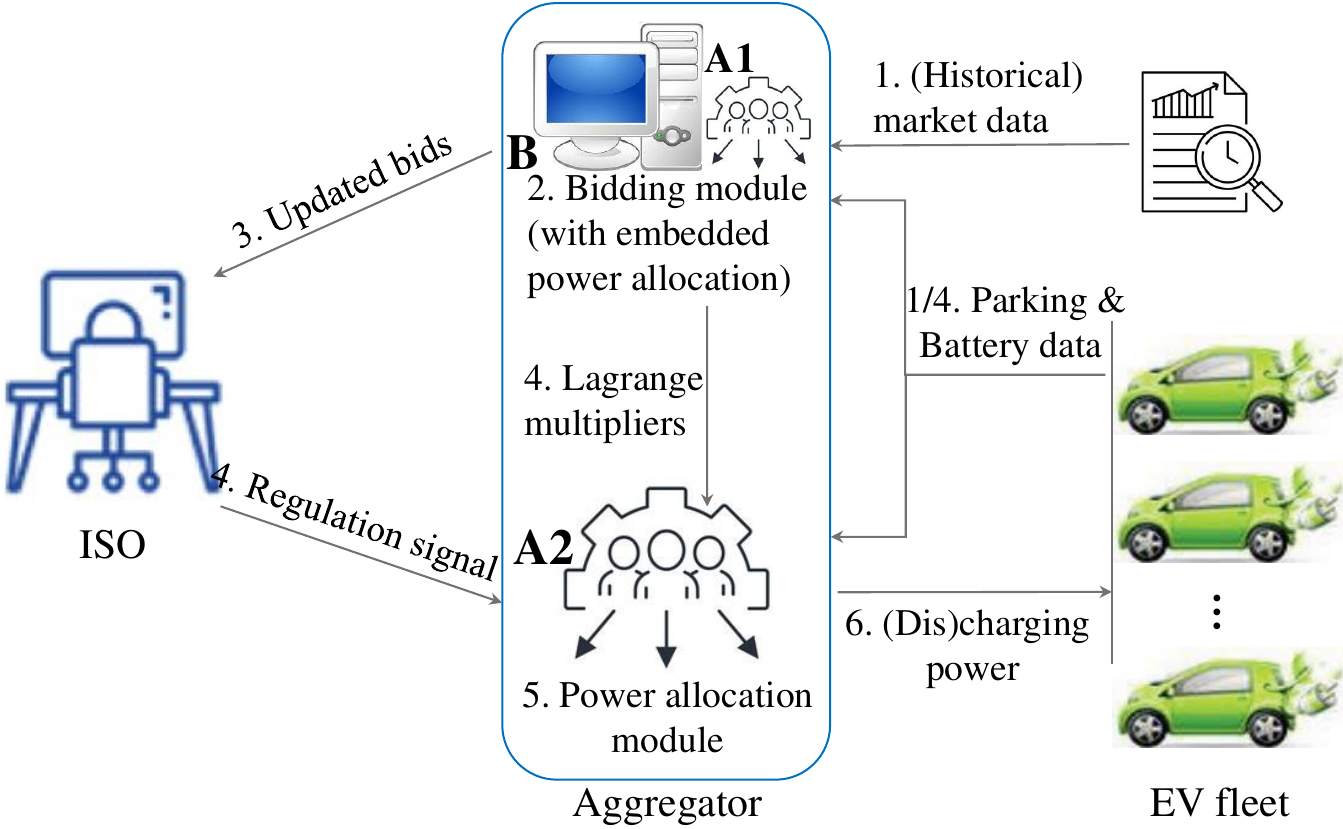}
  \caption{Overview of the proposed framework for EVA bidding and power allocation co-optimization.}
  \label{fig_overview}
\end{figure}
\begin{figure}[!t]
  \centering
  \includegraphics[width=2.8in]{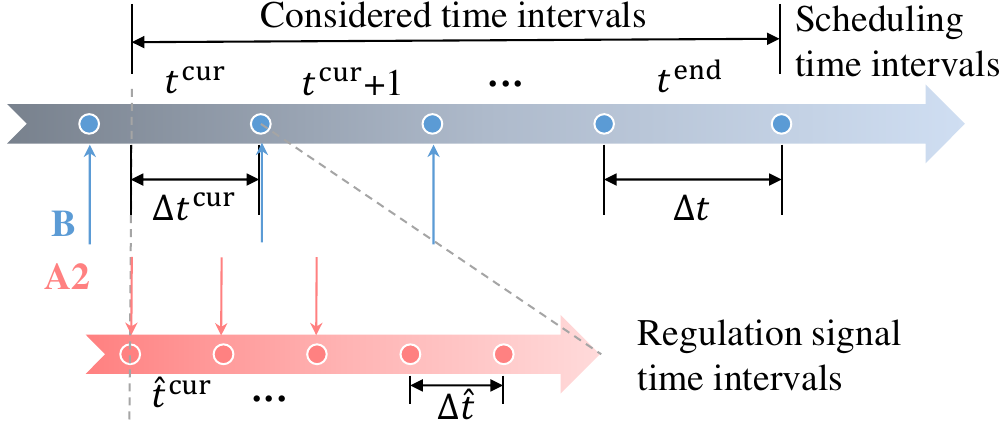}
  \caption{Timeline of the proposed framework: (B) Update the bids (periodically, e.g., every 30 minutes); (A2) Allocate the power for deploying regulation to EVs (upon receiving an AGC signal, e.g., every 2 seconds).}
  \label{fig_framework_time}
\end{figure}

The bidding module (B) runs periodically (e.g., every 30 min). The EVA collects the arriving and leaving time, current SOCs, charging needs, and battery parameters of the EVs and generates forecasts about prices and regulation signals according to historical market data. Based on these parameters, the bidding module outputs the latest optimal bids, and the Lagrange multipliers can be obtained without additional computational overhead. The power allocation submodule (A1) is embedded in the bidding module (B). Here, we divide the value region of regulation signals into discretized scenarios (see Fig.~\ref{fig_scenario} in the case study as a demonstration, where the granularity of the discretized regulation signal scenarios, i.e., the length of the value intervals is 0.1), and the power allocation in each scenario is simultaneously optimized.

Nevertheless, the power allocation results in the bidding model only correspond to the discretized regulation signal scenarios (e.g., $\delta_{\hat{s}} \in$ [0, 0.1), ..., [0.9, 1),...), while the actual regulation signal may take any value in its nearly continuous range (e.g., $\delta_{\hat{s}} = 0.9998$). Therefore, the bidding model cannot be directly used to respond to regulation signals. To address this problem, an intuitive way is to adopt a sufficiently small granularity of the discretized regulation signal scenarios, making them close to continuous (e.g., [0, 0.0001), ..., [0.9998, 0.9999), ...). However, the resulting large-scale problem cannot be solved online or is even computationally intractable. Instead, we design a power allocation module (A2) to quickly decide the (dis)charging power of the EVs in regulation deployment.

The power allocation module (A2) runs upon receiving the latest regulation signal issued by the ISO (e.g., every 2 s). Using the Lagrange multipliers from the bidding model as problem parameters, the power allocation problem outputs the (dis)charging power of the EVs to follow the regulation signal.

\section{Optimal Bidding Model with Embedded Power Allocation}\label{sec_model_bid}
In this section, we propose an optimal bidding model for an EVA to decide its bids in the real-time market, which addresses the cost of providing regulation with optimized power allocation. If not specified, the EVA uses the latest information.

\subsection{Profit Model}\label{sec_model_profit}
The EVA aims to maximize $Profit(t^{\rm cur})$, the sum of its profits for future time intervals from the first (current) time interval $t^{\rm cur}$ to the end $t^{\rm end}$, as formulated in (\ref{profit_total}).
\begin{equation}\label{profit_total}
  \begin{aligned}
    Profit(t^{\rm cur}) = & \underset{s \in S}{\Sigma}\pi_s
    \underset{t \in T}{\Sigma}(Income^{\rm e}_{s, t}
    + Income^{\rm cap}_{s, t}                                                        \\
                          & + Income^{\rm mil}_{s, t} + Income^{\rm ch}_{s, t}       \\
                          & + \underset{\hat{s} \in \hat{S}}{\Sigma}\pi_{t, \hat{s}}
    (Income^{\rm deploy}_{s, t,\hat{s}}
    - Cost^{\rm deg}_{t, \hat{s}}))
  \end{aligned}
\end{equation}
In Equation (\ref{profit_total}), the time intervals are set to be consistent with the scheduling time intervals of the market.
$T = \{t^{\rm cur}, t^{\rm cur}+1, ..., t^{\rm end}\}$ is the finite set of time intervals under consideration.

The specific components of the EVA profit are formulated in (\ref{profit_energy})-(\ref{cost_deg}), $\forall t \in T, s \in S, \hat{s} \in \hat{S}$:
\begin{flalign}\label{profit_energy} 
  &\ Income^{\rm e}_{s, t} =  Pr^{\rm e}_{s, t} \cdot P_{t}\cdot \Delta t
  &\\ \label{profit_cap} 
  &\ Income^{\rm cap}_{s, t} =  Pr^{\rm cap}_{s, t} \cdot R_{t} \cdot s^{\rm perf}\cdot \Delta t
  &\\ \label{profit_mil}
  &\ Income^{\rm mil}_{s, t} =  Pr^{\rm mil}_{s, t} \cdot R_{t} \cdot a^{\rm mil}_{t} \cdot s^{\rm perf}\cdot \Delta t
  &\\ \label{profit_deploy}
  &\ Income^{\rm deploy}_{s, t,\hat{s}} =  R_{t} \cdot \delta_{\hat{s}} \cdot Pr^{\rm e}_{s, t}\cdot \Delta t
  &\\ \label{cost_deg}
  &\ Cost^{\rm deg}_{t,\hat{s}} =   \underset{i \in I}{\Sigma} \underset{k \in K}{\Sigma} Pr^{\rm deg}_{i, k} \cdot P^{\rm dis}_{t, \hat{s}, i, k} \cdot \Delta t &
\end{flalign}

The EVA obtains income from the electricity market by energy arbitrage (\ref{profit_energy}) and providing regulation services. The income from the regulation market includes the payments for deploying regulation (\ref{profit_deploy}) (in which the exchanged energy is settled at real-time energy prices), the regulation capacity (\ref{profit_cap}), and the regulation mileage (\ref{profit_mil}).
In (\ref{profit_cap})-(\ref{profit_mil}),
$s^{\rm perf}$ is the performance score for providing the regulation mileage, and $a^{\rm mil}_{t}$ is the resource-specific mileage multiplier for regulation at time interval t. $s^{\rm perf}$ and $a^{\rm mil}_{t}$ are specified periodically by the ISO~\cite[Section 4.5.5]{bjm2022bpm12} and regarded as constants known by the EVA.
In (\ref{profit_deploy}),
$\delta_{\hat{s}}$ is the normalized value of the regulation signal in scenario $\hat{s}$, and a demonstration of the specific values is given in Fig.~\ref{fig_scenario} in the case study.
We regard injecting energy into the grid as the positive direction of energy and regulation signals.
As mentioned above, the income for charging EVs is a constant in the objective function that is irrelevant to the bidding decision. Therefore, we omit the expression of $Income^{\rm ch}_{s, t}$ in the model, although it does affect the profit of the EVA. 
In (\ref{cost_deg}), we use the degradation model for EV batteries in~\cite{han_practical_2014}, where the battery degradation is based on the depth of discharge (DOD) and can be linearly approximated by discharge segments. Equivalently, we express the battery degradation by segment-based discharge power, which is equivalent to segment-based discharge for a given length of time.
Note that due to the increasing marginal degradation cost of EV batteries, EVs will not charge and discharge at the same time as the EVA maximizes its profit, and therefore, 0-1 variables representing charge/discharge states of power segments need not be introduced.

In the above the EVA profit model, the battery degradation cost represents the compensation for EV users, and the EVA can adopt different profit distribution methods and modify the battery degradation cost item in the model accordingly.

\subsection{Constraints}\label{sec_model_constraint}
In the bidding model, the EVA needs to meet the technical and charging need constraints of the EVs (\ref{cons_positiveCapacity})-(\ref{cons_energyLimit}).
\begin{flalign}\label{cons_positiveCapacity}
  &\ 0 \le R_{t}, \ \forall t \in T
  &\\
  &\ P_{t} + R_{t}\cdot \delta_{\hat{s}} - \underset{i \in I}{\Sigma} \underset{k \in K}{\Sigma} (P^{\rm dis}_{t, \hat{s}, i, k} - P^{\rm ch}_{t, \hat{s}, i, k}) = 0 : \lambda^{\rm bal}_{t, \hat{s}},
  \nonumber
  &\\ \label{cons_balance}
  &\ \ \forall t \in T, \hat{s} \in \hat{S}
  &\\
  &\ \underset{i \in I}{\Sigma} (E_{t, i} - E^{\rm min}_{t, i})\cdot \eta^{\rm dis}_i \ge  R_{t} \cdot \Delta t^{\rm req} +  P_{t} \cdot \Delta t,
  \nonumber
  &\\ \label{cons_reg1}
  &\ \ \forall t \in T
  &\\
  &\ \underset{i \in I}{\Sigma} (- E_{t, i} + E^{\rm max}_{t, i}) / \eta^{\rm ch}_i \ge  R_{t} \cdot \Delta t^{\rm req} -  P_{t} \cdot \Delta t,
  \nonumber
  &\\ \label{cons_reg2}
  &\ \ \forall t \in T
  &\\
  &\ 0 \le P^{\rm dis(ch)}_{t, \hat{s}, i, k} \le u_{t, i} \cdot \overline{P^{\rm dis(ch)}_{i, k}} : \underline{\mu^{\rm dis(ch)}_{t, \hat{s}, i, k}}, \overline{\mu^{\rm dis(ch)}_{t, \hat{s}, i, k}},
  \nonumber
  &\\ \label{cons_powerLimit}
  &\ \ \forall  t \in T, \hat{s} \in \hat{S}, i \in I, k \in K
  &\\ \label{cons_energyInit}
  &\ E^{\rm cur}_{i} - E_{t, i} = 0, \ \forall i \in I, t = t^{\rm cur}
  &\\
  &\ E_{t, i} +  \underset{\hat{s} \in \hat{S}}{\Sigma}\pi_{t, \hat{s}} \underset{k \in K}{\Sigma} (P^{\rm ch}_{t, \hat{s}, i, k}\cdot \eta^{\rm ch}_i - P^{\rm dis}_{t, \hat{s}, i, k} / \eta^{\rm dis}_i) \cdot \Delta t 
  \nonumber
  &\\ \label{cons_energyChange}
  &\ \ - E_{t + 1, i} = 0 : \lambda^{\rm E}_{t, i},  \ \forall t \in T, i \in I
  &\\ \label{cons_energyLimit}
  &\ E^{\rm min}_{t, i} \le E_{t, i} \le E^{\rm max}_{t, i}, \ \forall t \in T, i \in I &
\end{flalign}

(\ref{cons_positiveCapacity}) states that the bids for regulation capacity are non-negative. The power balance constraint (\ref{cons_balance}) indicates that the net power of the EVs follows the regulation signal in each regulation signal scenario. Since the regulation signal $\delta_{\hat{s}} \in [-1, 1]$ can take the value of -1 and 1, (\ref{cons_balance}) also ensures that the fully deployed output power is satisfied.
Resources providing regulation are usually required by the ISO to be able to maintain the minimum and maximum output for some time $\Delta t^{\rm req}$, such as 15 min (PJM)~\cite{masiello2014business}, stated in (\ref{cons_reg1}) and (\ref{cons_reg2}), respectively.
The (dis)charging power limit at the power segments is expressed in (\ref{cons_powerLimit}), where $u_{t, i}$ is the plug-in state of EV $i$ at time
interval $t$, and $u_{t, i} = 1(0)$ represents being (not) connected to the charging slot.
The initial SOCs of the EVs, i.e., the SOCs at the beginning of the current time interval are stipulated by (\ref{cons_energyInit}), where $E^{\rm cur}_{i}$ is a problem parameter and $E_{t, i}$ $(t = t^{\rm cur})$ is a variable.
The expected change in SOC across the time intervals is given by (\ref{cons_energyChange}), where $E_{t, i}$ represents the SOC at the beginning of time interval $t$; 
for formal consistency, let $E_{t^{\rm end}+1, i}$, i.e., $E_{t, i}$ $(t = t^{\rm end}+1)$ denote the SOC of EV $i$ at the end of time interval $t^{\rm end}$.
The SOC limit of the EVs is expressed by (\ref{cons_energyLimit}).
$\lambda^{\rm bal}_{t, \hat{s}}$,
$\underline{\mu^{\rm dis(ch)}_{t, \hat{s}, i, k}}$,
$\overline{\mu^{\rm dis(ch)}_{t, \hat{s}, i, k}}$, and
$\lambda^{\rm E}_{t + 1, i}$ are the Lagrange multipliers of the corresponding constraints.
Following convention, we adopt Lagrange multipliers with the objective function written as min. and the inequalities written as $\le 0$, although they are expressed in the above form for convenience.

For the sake of compactness, the constraints regarding the time intervals when EV $i$ has not arrived or already left are not listed separately, but these constraints can be expressed nonexplicitly by (\ref{cons_energyInit})-(\ref{cons_energyLimit}) through different problem parameters. For example, regarding the time intervals when EV $i$ has left, set $E^{\rm min}_{t, i}$ in (\ref{cons_energyLimit}) to the required minimum SOC, which states the charging need of the EV.

In the bidding model, the decision variables for bids are $\{P_{t}, R_{t} | t \in T\}$. The power
allocation submodule (A1) is embodied in the decision variables for power allocation (i.e., $\{P^{\rm dis(ch)}_{t, \hat{s}, i, k} | t \in T, \hat{s} \in \hat{S}, i \in I, k \in K\}$) and the corresponding constraints for each regulation signal scenario (i.e., (\ref{cost_deg})(\ref{cons_balance})(\ref{cons_powerLimit})(\ref{cons_energyChange})). 
Conceptually, the power allocation submodule embedded in the bidding module must correspond to the actual power allocation module in regulation deployment. We will show such correspondence in Section~\ref{sec_optimality}.
The subscripts of $P^{\rm dis(ch)}_{t, \hat{s}, i, k}$ indicate that the same power allocation strategy is applied regarding the same regulation signal scenario $\hat{s}$ at time interval $t$. This assumption is reasonable since the scheduling time interval in the real-time market is short enough, and the change in the problem parameters within $t$ is negligible.
The entire model, formulated as a stochastic linear programming problem, can be efficiently solved by commercial solvers.

By modifying the length of the current time interval in the corresponding constraints, the model can be used for bidding at any time within a time interval. Specifically, at time $\hat{t}^{\rm cur}$ within the current time interval $t^{\rm cur}$, i.e., $\hat{t}^{\rm cur} \in t^{\rm cur}$, replace the length of the time interval with the remaining time of $\Delta t^{\rm cur}$ in constraints (\ref{profit_energy})-(\ref{cost_deg}) and (\ref{cons_energyChange}) regarding $t^{\rm cur}$. For instance, in (\ref{cons_energyChange}) when $t = t^{\rm cur}$:
\begin{equation}\label{cons_energyChange_mid}
  \begin{aligned}
     & E_{t^{\rm cur}, i} +  \underset{\hat{s} \in \hat{S}}{\Sigma}\pi_{t^{\rm cur}, \hat{s}} \underset{k \in K}{\Sigma} (P^{\rm ch}_{t^{\rm cur}, \hat{s}, i, k}\cdot \eta^{\rm ch}_i - P^{\rm dis}_{t^{\rm cur}, \hat{s}, i, k} / \eta^{\rm dis}_i) \Delta t^{\rm cur} \\
     &  
    - E_{t^{\rm cur} + 1, i} = 0 : \lambda^{\rm E}_{t^{\rm cur}, i},  \ \forall i \in I
  \end{aligned}
\end{equation}
Meanwhile, the bids $P_{t^{\rm cur}}$ and $R_{t^{\rm cur}}$ for the current time interval usually can no longer be modified but should be regarded as problem parameters.

In the following, unless specified otherwise, we call the optimal bidding problem at time $\hat{t}^{\rm cur} \in t^{\rm cur}$ the bidding problem, represented as:
\begin{flalign}\label{optimal_bid_mid}
  &\ {\rm max.} \ Profit(\hat{t}^{\rm cur})&\\
\nonumber
  &\ {\rm s.t.} \ \ \ (\ref{profit_total})-(\ref{cons_energyLimit}), {\rm \ \ constraints \ \ regarding \ \ }t^{\rm cur} {\rm \ \ revised}&
\end{flalign}

For the convenience of presentation, the model formulated above is more general than detailed.
Here we provide a few possible modifications for specific situations:
\begin{itemize}
  \item The bids for the closed time intervals should take the latest declared value, regarded as the problem parameters instead of decision variables.
  \item If the market has restrictions or penalties for bid modification or power deviation, the model can be modified accordingly, as in~\cite{jin_optimizing_2013}.
  \item In our consideration of the connection point of EV, the plug-in state is known leveraging the contracts and the parking time granularity is consistent with time intervals, and we can also treat it as a random variable~\cite{yao_optimization_2018} or reduce its granularity~\cite{vagropoulos_optimal_2013}.
  \item The cost of deploying regulation (\ref{cost_deg}) can be modified if the EVA compensates EV owners in other ways.
  \item If the EVA is required to ensure that the power flow of the distribution network does not exceed the limit, linearized power flow constraints of the distribution network can be added to the model as in~\cite{vatandoust_risk-averse_2019}.
\end{itemize}

\section{Power Allocation Model Based on Lagrange Multipliers}\label{sec_model_allocation}

\subsection{Power Allocation Model}
At time $\hat{t}^{\rm cur} \in t^{\rm cur}$, upon receiving a regulation signal $\delta_{\hat{s}}$ from the ISO, the EVA runs the following power allocation problem (\ref{optimal_allocation})-(\ref{cons_energyLimit3}):
\begin{flalign}\label{optimal_allocation}
  &\ {\rm min.} \  Cost_{\hat{t}^{\rm cur}, \hat{s}} + M \cdot (\Delta P^{\rm dis}_{\hat{t}^{\rm cur}, \hat{s}} + \Delta P^{\rm ch}_{\hat{t}^{\rm cur}, \hat{s}})
  &\\ 
  &\ {\rm s.t.} \ \ \ Cost_{\hat{t}^{\rm cur}, \hat{s}} =   \underset{i \in I}{\Sigma} \underset{k \in K}{\Sigma} (Pr^{\rm deg}_{i, k} \cdot P^{\rm dis}_{\hat{t}^{\rm cur}, \hat{s}, i, k}
  \nonumber
  &\\ \label{cons_cost2}
  &\ \qquad  + \lambda^{\rm E}_{t^{\rm cur}, i} \cdot (P^{\rm ch}_{\hat{t}^{\rm cur}, \hat{s}, i, k}\cdot \eta^{\rm ch}_i - P^{\rm dis}_{\hat{t}^{\rm cur}, \hat{s}, i, k} / \eta^{\rm dis}_i)) \cdot \Delta \hat{t}
  &\\ 
  &\ \qquad P_{t^{\rm cur}} + R_{t^{\rm cur}}\cdot \delta_{\hat{s}} - \underset{i \in I}{\Sigma} \underset{k \in K}{\Sigma} (P^{\rm dis}_{\hat{t}^{\rm cur}, \hat{s}, i, k} - P^{\rm ch}_{\hat{t}^{\rm cur}, \hat{s}, i, k})
  \nonumber
  &\\ \label{cons_balance2}
  &\ \qquad - (\Delta P^{\rm dis}_{\hat{t}^{\rm cur}, \hat{s}} - \Delta P^{\rm ch}_{\hat{t}^{\rm cur}, \hat{s}}) = 0 : \lambda^{\rm bal}_{\hat{t}^{\rm cur}, \hat{s}}
  &\\ 
  &\ \qquad  0 \le P^{\rm dis(ch)}_{\hat{t}^{\rm cur}, \hat{s}, i, k} \le u_{t^{\rm cur}, i} \cdot \overline{P^{\rm dis(ch)}_{i, k}} : \underline{\mu^{\rm dis(ch)}_{\hat{t}^{\rm cur}, \hat{s}, i, k}}, \overline{\mu^{\rm dis(ch)}_{\hat{t}^{\rm cur}, \hat{s}, i, k}},
  \nonumber
  &\\ \label{cons_powerLimit2}
  &\  \qquad \forall i \in I, k \in K
  &\\ \label{cons_powerLimit_unbal}
  &\  \qquad 0 \le \Delta P^{\rm dis(ch)}_{\hat{t}^{\rm cur}, \hat{s}}
  &\\ 
  &\  \qquad E^{\rm min}_{t^{\rm cur}, i} \le E^{\rm cur}_{i} +   \underset{k \in K}{\Sigma} (P^{\rm ch}_{\hat{t}^{\rm cur}, \hat{s}, i, k}\cdot \eta^{\rm ch}_i
  \nonumber
  &\\ \label{cons_energyLimit2}
  &\  \qquad - P^{\rm dis}_{\hat{t}^{\rm cur}, \hat{s}, i, k} / \eta^{\rm dis}_i) \cdot \Delta \hat{t} \le E^{\rm max}_{t^{\rm cur}, i}, \forall i \in I
  &\\ 
  &\ \qquad  E^{\rm leave}_{i} - \underset{k \in K}{\Sigma} \overline{P^{\rm ch}_{i, k}}\cdot \eta^{\rm ch}_i \cdot t^{rest}_{i} \le E^{\rm cur}_{i}
  \nonumber
  &\\ \label{cons_energyLimit3}
  &\  \qquad +   \underset{k \in K}{\Sigma} (P^{\rm ch}_{\hat{t}^{\rm cur}, \hat{s}, i, k}\cdot \eta^{\rm ch}_i
  - P^{\rm dis}_{\hat{t}^{\rm cur}, \hat{s}, i, k} / \eta^{\rm dis}_i) \cdot \Delta \hat{t}, \forall i \in I &
\end{flalign}

$M$ is a large number, and the second term in the objective function (\ref{optimal_allocation}) is a penalty term, which forces the EVA to follow the regulation signal as closely as possible. 
The cost is calculated from the battery degradation cost that is immediately generated by EV discharge in addition to the impact on future EVA profit accounted for through the Lagrange multiplier $\lambda^{\rm E}_{t^{\rm cur}, i}$, which will be analyzed in Section~\ref{sec_optimality}. 
The power deviated from the regulation signal is expressed in (\ref{cons_balance2}). (\ref{cons_energyLimit2}) and (\ref{cons_energyLimit3}) ensure that the SOCs of EVs will be acceptable after the power allocation at $\hat{t}$, where (\ref{cons_energyLimit3}) guarantees that EV $i$ can be charged to the required minimum SOC at departure in the worst-case scenario. Note that in the bidding problem, the EV SOC limit (\ref{cons_energyLimit}) is satisfied only in the mathematical expectation, but (\ref{cons_energyLimit2}) and (\ref{cons_energyLimit3}) here ensure the real-world feasibility of power allocation. $\lambda^{\rm bal}_{\hat{t}^{\rm cur}, \hat{s}}$,
$\underline{\mu^{\rm dis(ch)}_{\hat{t}^{\rm cur}, \hat{s}, i, k}}$,
$\overline{\mu^{\rm dis(ch)}_{\hat{t}^{\rm cur}, \hat{s}, i, k}}$ are the Lagrange multipliers of the corresponding constraints.

We call the above problem (\ref{optimal_allocation}) the power allocation problem. The decision variables are $\{P^{\rm dis(ch)}_{\hat{t}^{\rm cur}, \hat{s}, i, k} | i \in I, k \in K\}$, and the power allocation at different time intervals or regulation signal scenarios need not be considered in (\ref{optimal_allocation}). Therefore, the scale of the power allocation problem (in the form of linear programming) is much smaller than that of the bidding problem (\ref{optimal_bid_mid}) and can be solved online. Nevertheless, the optimality of the power allocation result needs to be guaranteed, and we will analyze it below.

\subsection{Optimality of Power Allocation}\label{sec_optimality}

In this paper, the optimality of the power allocation problem refers to the fact that the optimal solution to the power allocation problem (\ref{optimal_allocation}) $\{P^{\rm *dis(ch)}_{\hat{t}^{\rm cur}, \hat{s}, i, k} | i \in I, k \in K\}$ corresponds to the optimal solution to the bidding problem (\ref{optimal_bid_mid}) $\{P^{\rm *dis(ch)}_{t^{\rm cur}, \hat{s}, i, k} | i \in I, k \in K\}$, $\forall \hat{s} \in \hat{S}$.

Intuitively, the Lagrange multiplier $\lambda^{\rm E}_{t^{\rm cur}, i}$ of the SOC constraint (\ref{cons_energyChange_mid}) represents the shadow price of $E_{t^{\rm cur}, i}$ for the profit of EVA bidding, i.e., the cost of increasing a unit SOC of EV $i$ at $t^{\rm cur}$.
By considering the impact on future EVA profit through $\lambda^{\rm E}_{t^{\rm cur}, i}$ together with the immediately induced degradation cost, the solution to the power allocation problem will also maximize EVA profit.

In the mathematical meaning, introducing Lagrange multipliers of the bidding problem into the power allocation problem makes the optimality conditions of the two problems correspond. In the remainder of this section, we theoretically prove the optimality of the power allocation problem. Our analysis is based on the following assumptions:
\begin{assumption}\label{assumption_feasible}
  The bidding problem (\ref{optimal_bid_mid}) and the power allocation problem (\ref{optimal_allocation}) are strictly feasible.
\end{assumption}
\begin{assumption}\label{assumption_express}
  The actual regulation signals can be expressed by the regulation signal scenarios, i.e., $\hat{s} \in \hat{S}$ and $\pi_{t, \hat{s}} > 0, \forall t \in T$, $\forall \delta_{\hat{s}} \in [-1, 1]$.
\end{assumption}
\begin{assumption}\label{assumption_forecast}
  The occurrence frequency of the regulation signals is consistent with the forecasted probability.
\end{assumption}
Due to the flexibility of charging, Assumption~\ref{assumption_feasible} is realistic. Assumption~\ref{assumption_express} does not hold because the value region of regulation signals is closer to continuous rather than discrete. Nevertheless, we examined the applicability of Assumption~\ref{assumption_express} with suitable granularity in the case study. Assumption~\ref{assumption_forecast} indicates that the SOC constraints (\ref{cons_energyLimit}) in the bidding model are met not only on the expectation but also in actual deployment.

\begin{theorem}\label{theorem_optimality}
  \textbf{Optimality.} If the Lagrange multipliers in (\ref{cons_cost2}) correspond to the optimal solution of the bidding problem, then the optimal solution to the power allocation problem for $\delta_{\hat{s}}$ at $\hat{t}^{\rm cur} \in t^{\rm cur}$ corresponds to the optimal solution to the bidding problem for $\delta_{\hat{s}}$ at $t^{\rm cur}$.
\end{theorem}
\begin{proof}
  See Appendix~\ref{theorem_optimality_app}.
\end{proof}

\begin{theorem}\label{theorem_invariance}
  \textbf{Time invariance.} If starting from time $\hat{t}^{\rm cur} \in t^{\rm cur}$, for each regulation signal $\delta_{\hat{s}}$, allocate the power of the EVs according to the optimal solution of the bidding problem until time $\hat{t}'^{\rm cur} \in t^{\rm cur}$, $\hat{t}'^{\rm cur} > \hat{t}^{\rm cur}$, then the optimal solution to ${\rm max.} \ Profit(\hat{t}^{\rm cur})$ is also optimal to ${\rm max.} \ Profit(\hat{t}'^{\rm cur})$.
\end{theorem}
\begin{proof}
  See Appendix~\ref{theorem_invariance_app}.
\end{proof}

Theorem~\ref{theorem_optimality} and Theorem~\ref{theorem_invariance} together indicate that during a specific time interval, the optimal bidding problem only needs to be solved once, and the obtained Lagrange multipliers can be used as the parameters of the power allocation problem without updating.
However, the forecasts for prices and regulation signals may be inaccurate and time-varying, and Theorem~\ref{theorem_invariance} does not hold. To cope with this, a practical method is to solve the bidding problem periodically and update the bids and the Lagrange multipliers (Algorithm~\ref{alg_framework}).

\begin{algorithm}[!t]
  \caption{Co-optimizing Bidding and Power Allocation.}
  \label{alg_framework}
  \begin{algorithmic}[1]
    \renewcommand{\algorithmicrequire}{\textbf{Input:}}
    \REQUIRE
    market data, EV parking and battery data
    \renewcommand{\algorithmicrequire}{\textbf{Output:}}
    \REQUIRE Updated bids $\{P_{t}, R_{t}\}$ and EV power allocation for any AGC signal $\{P^{\rm dis(ch)}_{\hat{t}^{\rm cur}, \hat{s}, i, k}\}$.
    \renewcommand{\algorithmicensure}{\textbf{B. Optimal bidding with embedded power allocation}}
    \ENSURE (Triggered periodically, e.g., every 30 min) \
    \STATE generate forecasted market prices $\{Pr^{\rm e/cap/mil}_{s, t}\}$
    \STATE generate forecasted distribution of AGC signals $\{\pi_{t, \hat{s}}\}$
    \STATE optimize the bids $\{P_{t}, R_{t}\}$ and the power allocation in the discretized AGC signal scenarios $\{P^{\rm dis(ch)}_{t, \hat{s}, i, k}\}$ via (\ref{optimal_bid_mid})
    \STATE update the bids to the market and the Lagrangian multipliers $\{\lambda^{\rm E}_{t^{\rm cur}, i}\}$ to \textbf{A2}.
    \renewcommand{\algorithmicensure}{\textbf{A2. Power allocation based on Lagrangian multipliers}}
    \ENSURE  (Triggered by receiving an AGC signal, e.g., every 2 s) \
    \STATE optimize the power allocation regarding the received AGC signal $\{P^{\rm dis(ch)}_{\hat{t}^{\rm cur}, \hat{s}, i, k}\}$ via (\ref{optimal_allocation})
    \STATE update the (dis)charging power to the EVs and the SOCs of the EVs to \textbf{B}.
  \end{algorithmic}
\end{algorithm}

\section{Case Study}\label{sec_case_study}
We used CPLEX (V12.4) and MATLAB (R2021a) with YALMIP~\cite{Lofberg2004} to solve the optimization problems on a workstation with an Intel Core i9-10900X CPU (3.7 GHz) and 128 GB RAM.

\subsection{Basic Data and Assumptions}

\begin{table*}[!t]
  \renewcommand{\arraystretch}{1.0}
  \caption{Parameters of the Three Types of EVs}
  \label{tab_ev_parameter}
  \centering
  \begin{tabular}{cccccccccc}
    \toprule
    Type                       &
    \begin{tabular}[c]{@{}c@{}}Discharging \\Power Limit   \\   (kW)\end{tabular} &
    \begin{tabular}[c]{@{}c@{}}Charging \\Power Limit   \\   (kW)\end{tabular} &
    \begin{tabular}[c]{@{}c@{}}Initial \\ SOC    \\ (kWh)\end{tabular} &
    \begin{tabular}[c]{@{}c@{}}Required \\   SOC  \\ (kWh)\end{tabular} &
    \begin{tabular}[c]{@{}c@{}}Lower  \\  SOC Limit \\  (kWh)\end{tabular} &
    \begin{tabular}[c]{@{}c@{}}Upper  \\ SOC Limit  \\  (kWh)\end{tabular} &
    \begin{tabular}[c]{@{}c@{}}Discharging \\  Efficiency  \end{tabular} &
    \begin{tabular}[c]{@{}c@{}}Charging \\  Efficiency  \end{tabular} &
    \begin{tabular}[c]{@{}c@{}}Degradation \\  Factor  \\  (\$/kWh)\end{tabular}                                                        \\
    \midrule
    a                          & 7.68 & 7.68 & 30 & 81 & 15 & 90 & 0.95 & 0.95 & 0.10 \\
    b                          & 7.68 & 7.68 & 20 & 54 & 10 & 60 & 0.92 & 0.92 & 0.15 \\
    c                          & 7.68 & 7.68 & 12 & 32 & 6  & 36 & 0.90 & 0.90 & 0.20 \\
    \bottomrule
  \end{tabular}
\end{table*}

The proposed framework was investigated in the nighttime slow charging scenario with 402 EVs, and the charging level was set as the widely used SAE J1772 Level 2~\cite{J1772_201710}.
Three types of EVs (134 EVs for each type) were considered. The parameters are listed in Table~\ref{tab_ev_parameter}.
Considering the low power limit of slow charging and the shallow depth of discharge when providing regulation~\cite{he_optimal_2016}, we assumed that the battery degradation cost was proportional to the net discharged electricity.
Batteries of different types and qualities can feature different degradation characteristics~\cite{fang_efficient_2022}.
Adopting the EV battery degradation model in~\cite{han_practical_2014} with the updated battery price in~\cite{bloomberg2021electric}, the degradation factors of the EVs (i.e., degradation cost per discharged electricity) were set between \$0.10-0.20/kWh.

A larger degradation factor was set for Type-c EVs, considering that the battery with a smaller capacity will face a larger DOD and therefore more degradation when discharged the same amount of electricity.

The parking patterns of the EVs were generated according to the distribution of the arrival and departure time in~\cite{li_emission-concerned_2013}, except that the EVs that are always plug-in are treated as arriving at the beginning and leaving at the end so that the existing methods for comparison apply.

\begin{figure}[!t]
  \centering
  \includegraphics[width=3.5in]{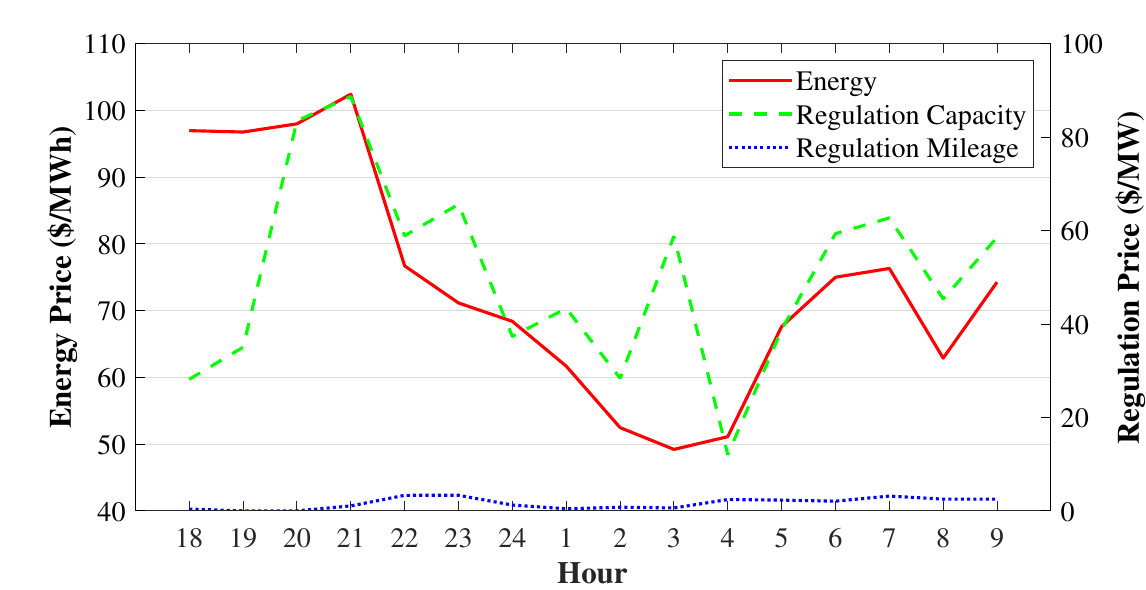}
  \caption{Real-time market prices (17.07.2022 18:00 to 18.07.2022 9:00).}
  \label{fig_price}
\end{figure}

We used historical market data in July 2022 from PJM.
For demonstration purposes, we assumed that the prices can be accurately forecasted while the regulation-signal-related parameters in the bidding problem are forecasted results.
The hourly real-time system energy prices and regulation prices from 17.07.2022 18:00 to 18.07.2022 9:00 were used as the real-time market prices (Fig.~\ref{fig_price}).
We also assumed that before entering a certain time interval, the bids regarding the time interval can be modified, although the closing time is earlier in some markets.
Other market parameters included the following:
the length of (scheduling) time interval $\Delta t = 1$ h,
the required maintenance time at the fully deployed output level $\Delta t^{\rm req} = 0.25$ h,
the performance score $s^{\rm perf}$ set to be 0.984 which is the same as the performance score example 1 given by PJM in~\cite{example_pjm}
and the hourly mileage multiplier $a^{\rm mil}_{t}$ calculated using historical RegD signals of the same hour in the past 14 days. 
The EVA solved the bidding problem starting at 17:59 and updated its bids and Lagrange multipliers with a default period of 30 min.

\begin{figure}[!t]
  \centering
  \includegraphics[width=3.5in]{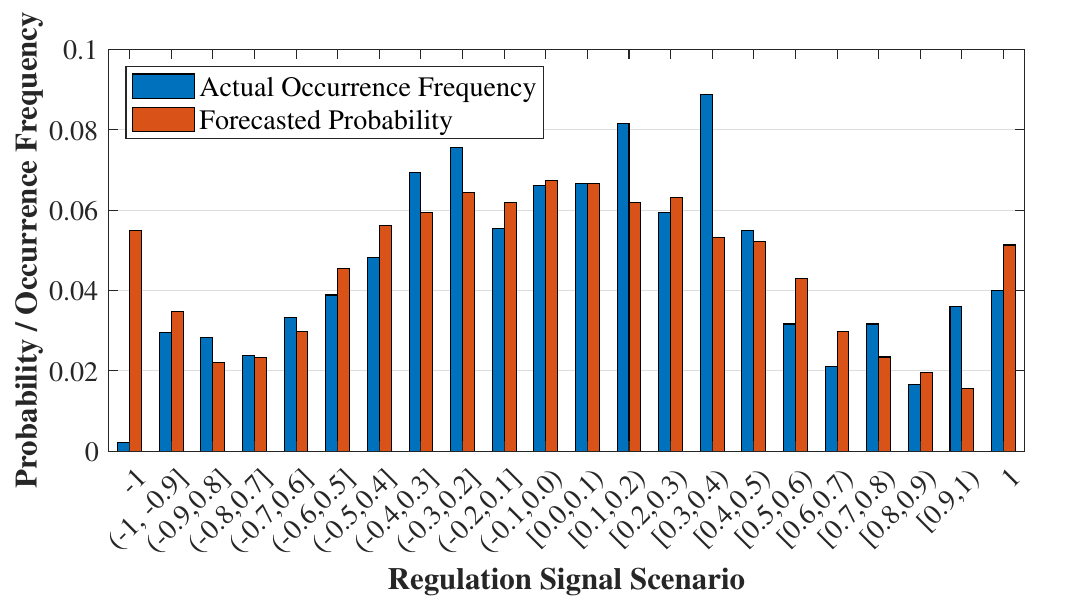}
  \caption{Distribution of regulation signal scenarios at 6:00-7:00.}
  \label{fig_scenario}
\end{figure}

In the simulation of regulation deployment, historical RegD signal data from PJM starting from 19.07.2020 18:00 was used, as this is the latest available data corresponding to the day type of 17.07.2022 (Sunday).
To estimate the probabilities and mileage multipliers for regulation signals at each time interval, we used the average distribution of RegD signals and the average mileage multipliers over the 14 days before the day of scheduling for the corresponding time intervals.
The default granularity of regulation signal scenarios was 0.1, and the interval median was used as the signal value of the corresponding scenario, e.g., $\delta_{\hat{s}} = 0.95$ for $\hat{s} = [0.9, 1)$. The forecasted and actual distributions of the regulation signal scenarios at 6:00-7:00 are shown in Fig.~\ref{fig_scenario} as a demonstration.

\subsection{Results and Comparison}
The methods being compared are combinations of different bidding strategies and power allocation strategies. Apart from the proposed method, the other three methods use the same bidding strategy as~\cite{vagropoulos_optimal_2013}, but their power allocation strategies in regulation deployment are different. For the sake of presentation, we name the compared methods as follows:
\paragraph*{\textbf{The proposed method}} Power allocation is embedded in the bidding model (\ref{optimal_bid_mid}), and the Lagrange-multiplier-based power allocation model (\ref{optimal_allocation}) decides the power of the EVs in regulation deployment.
\paragraph*{\textbf{Proportional Allocation}} In the bidding model, the EVA regulation capacity is allocated to the EVs, and the adjusted power of the EVs in regulation deployment is set proportional to the allocated capacity~\cite{vagropoulos_optimal_2013}.
\paragraph*{\textbf{Heuristic Weight}} Bidding as in~\cite{vagropoulos_optimal_2013}; degradation cost and charging priority are considered in power allocation. This is equivalent to replacing the Lagrange multipliers in (\ref{cons_cost2}) with the heuristic weights proposed in~\cite{vagropoulos_real-time_2016}. The recommended parameters are used, i.e., $(x_1, x_2, y, z) = (1, 1, 4, 5)$.
\paragraph*{\textbf{Minimum Degradation}} Bidding as in~\cite{vagropoulos_optimal_2013}; power is preferentially allocated to EVs with lower degradation factors. This is equivalent to setting the Lagrange multipliers in (\ref{cons_cost2}) to be $0$.

\begin{figure}[!t]
  \centering
  \includegraphics[width=3.5in]{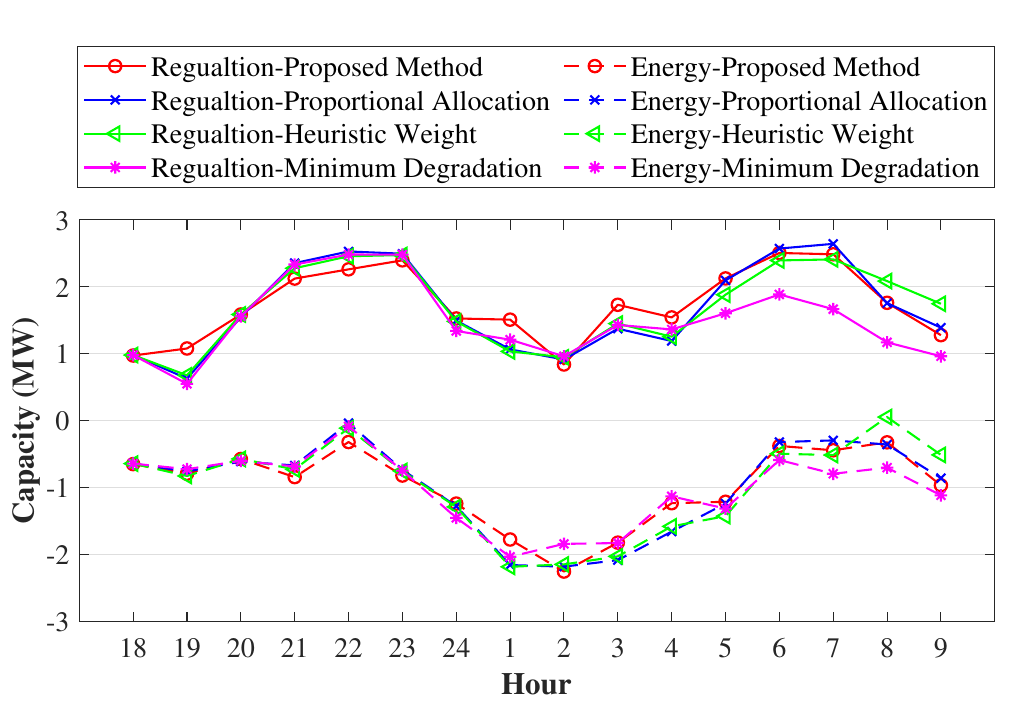}
  \caption{Final bids of the EVA in the electricity market.}
  \label{fig_bids}
\end{figure}

\begin{figure*}[!t]
  \centering
  \subfloat[]{
    \includegraphics[width=1.75in]{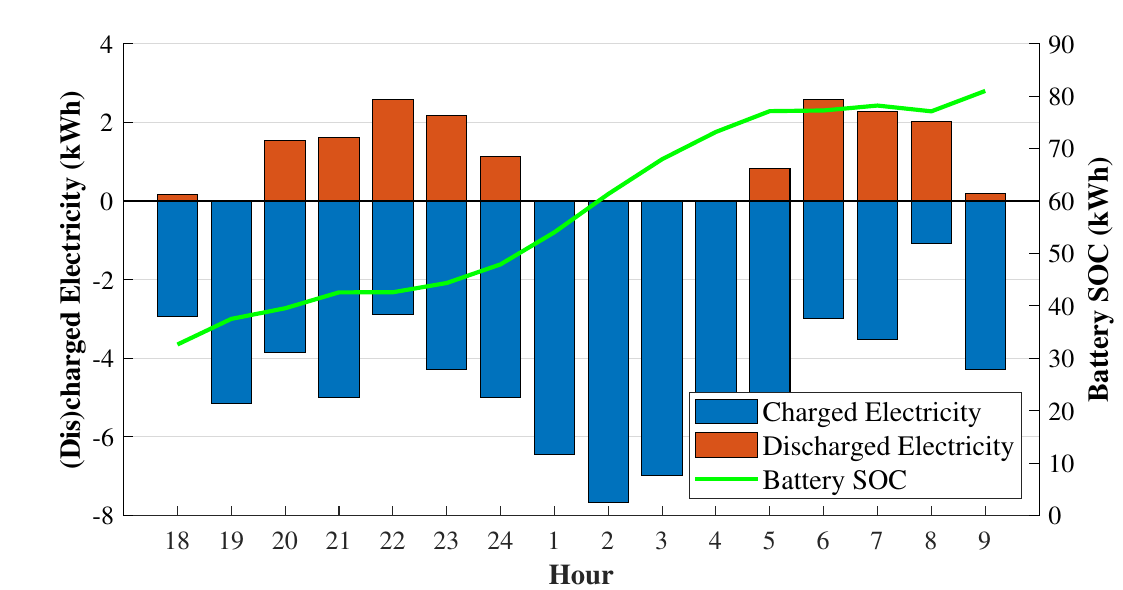}}
  \subfloat[]{
    \includegraphics[width=1.75in]{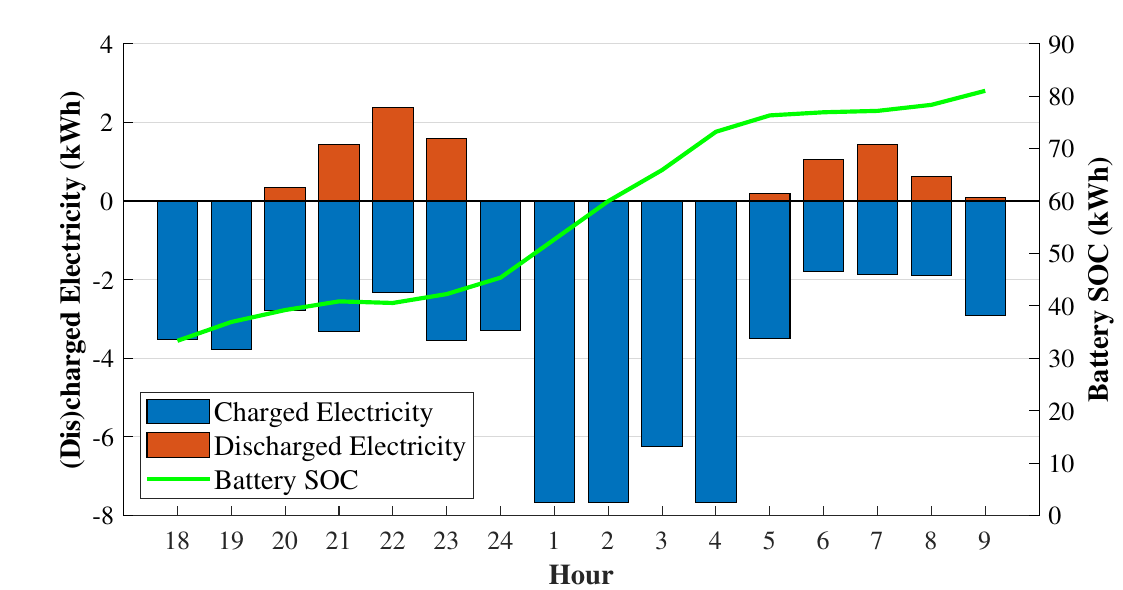}}
  \subfloat[]{
    \includegraphics[width=1.75in]{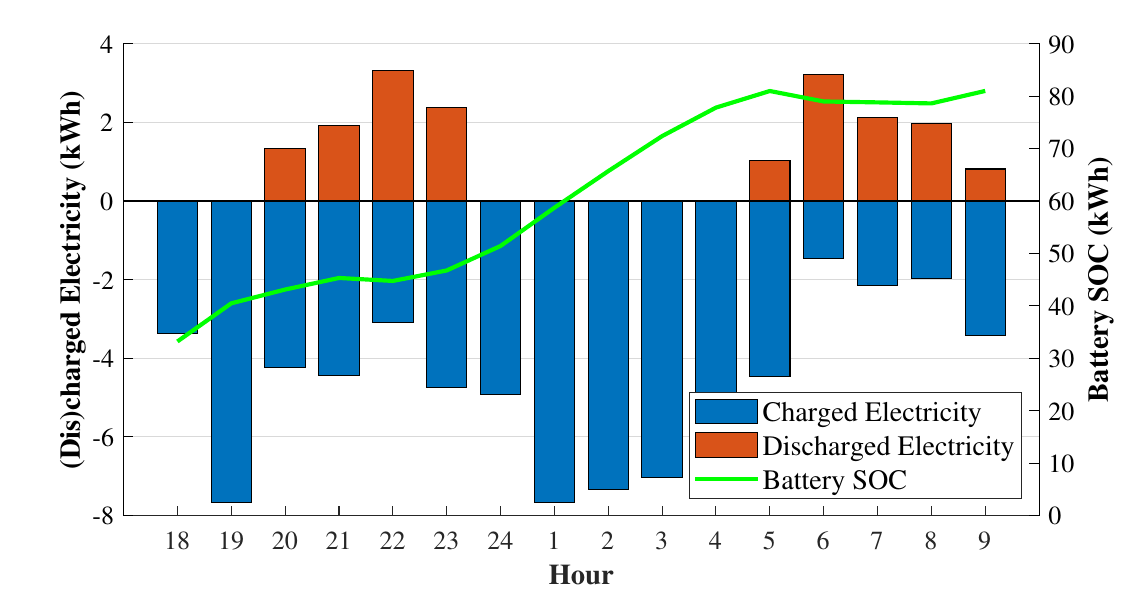}}
  \subfloat[]{
    \includegraphics[width=1.75in]{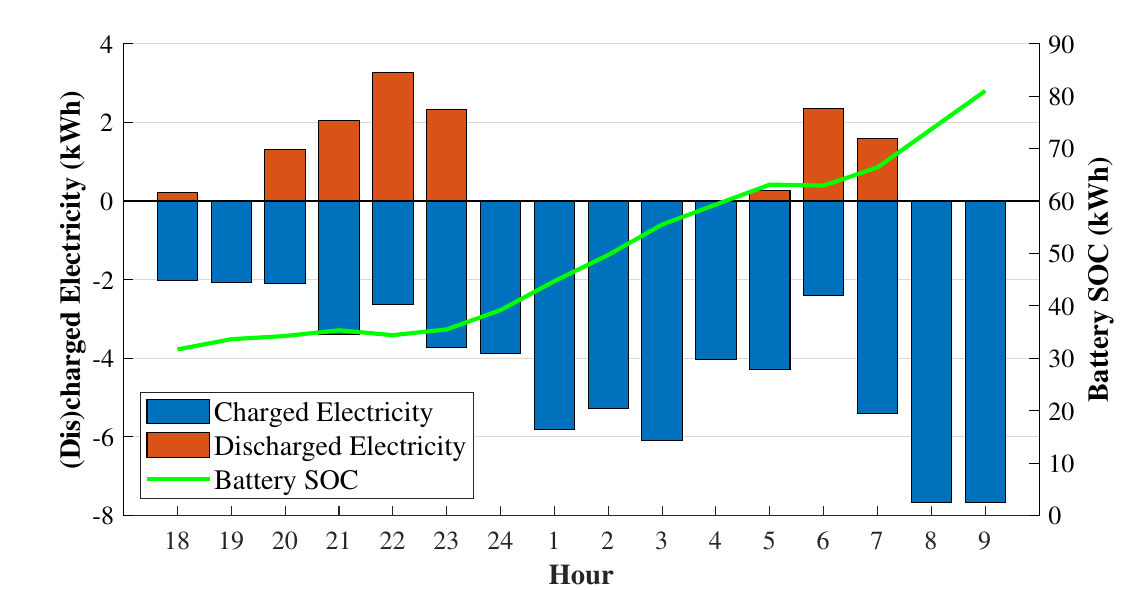}}
  \caption{Exchanged electricity with the grid and energy Curve of EV211: (a) The proposed method. (b) Proportional allocation. (c) Heuristic weight. (d) Minimum degradation.}
  \label{fig_typical_ev}
\end{figure*}

\begin{table}[!t]
  \caption{Comparison of EVA Profits Using Different Methods}
  \label{tab_profit_method}
  \centering
  \begin{tabular}{cccc}
    \toprule
    \multicolumn{1}{c}{\multirow{2}{*}{\begin{tabular}[c]{@{}c@{}}Bidding-Allocation \\      Method\end{tabular}}} &
    \multicolumn{1}{c}{\begin{tabular}[c]{@{}c@{}}Income from\\      the Market\end{tabular}}                  &
    \multicolumn{1}{c}{\begin{tabular}[c]{@{}c@{}}Degradation\\      Cost\end{tabular}}                  &
    \multicolumn{1}{c}{\begin{tabular}[c]{@{}c@{}}EVA \\      Profit\end{tabular}}                                                                                                   \\
    \multicolumn{1}{c}{}                                            & \multicolumn{1}{c}{(\$)} & \multicolumn{1}{c}{(\$)} & \multicolumn{1}{c}{(\$)} \\ \midrule
    \multirow{2}{*}{\textbf{Proposed Method}}                       & \textbf{1602}            & \textbf{319}             & \textbf{1283}            \\
                                                                    & \textbf{(-1.5\%)}        & \textbf{(-39.7\%)}       & \textbf{(+16.9\%)}       \\ \cline{2-4}
    \multirow{2}{*}{Proportional Allocation}                        & 1626                     & 528                      & 1098                     \\
                                                                    & 0\%                      & 0\%                      & 0\%                      \\\cline{2-4}
    \multirow{2}{*}{Heuristic Weight}                               & 1625                     & 444                      & 1181                     \\
                                                                    & (-0.05\%)                 & (-16.0\%)                & (+7.6\%)                 \\\cline{2-4}
    \multirow{2}{*}{Minimum Degradation}                            & 1200                     & 442                      & 759                      \\
                                                                    & (-26.2\%)                & (-16.4\%)                & (-30.9\%)                \\ \bottomrule
  \end{tabular}
\end{table}

\textbf{EVA Bids and Profits.} Fig.~\ref{fig_bids} and Table~\ref{tab_profit_method} show the final bids and profits of the EVA using each method.
Clearly, different power allocation strategies can lead to differing bids.
As shown in Table~\ref{tab_profit_method}, regarding the method of proportional allocation as the baseline, using the proposed method reduced the degradation cost of EV batteries by 39.7\%, thus increasing profit by 16.9\%.
Using the heuristic weight method had less impact. 
Although using the minimum degradation method reduced the degradation cost by 16.4\%, the corresponding profit decreased (-30.9\%) due to a greater decline in the market income (-26.2\%). It is worth noting that although the minimum degradation method attempts to minimize the immediate degradation cost upon receiving each regulation signal, the total degradation cost cannot be guaranteed to be the lowest because the long-term cost for deploying regulation is not considered. We will analyzed the in-depth reason in the discussion section.

\textbf{Energy Curve of a Typical EV.} Figs.~\ref{fig_typical_ev}(a-d) show the amount of electricity exchanged with the grid and the SOC at each time interval of a typical EV (No. 211, Type a, arrived at 18:00 on the 17th and left at 10:00 on the 18th).
Comparing Fig.~\ref{fig_typical_ev} (a) and (b), the discharged electricity of EV211 using the proposed method is much larger than that using the proportional allocation method, especially at hours 5 to 8.
Since the difference in the bids using the two methods is smaller (as shown in Fig.\ref{fig_bids} above), this indicates that EVs with lower degradation factors are allocated more discharging power using the proposed method.
In Fig.~\ref{fig_typical_ev}(d), since the minimum degradation method only considers the degradation cost at the current time, EV211 discharges a large amount of electricity at hours 6 and 7.
In the following hours 8 and 9, the EV keeps charging at the maximum power to meet the required SOC before leaving and therefore fails to provide more regulation capacity, resulting in a decline in the EVA income.
The heuristic weight method (Fig.~\ref{fig_typical_ev}(c)) partially reconciles the aim of reducing degradation costs and increasing future income.

\begin{figure}[!t]
  \centering
  \includegraphics[width=3.5in]{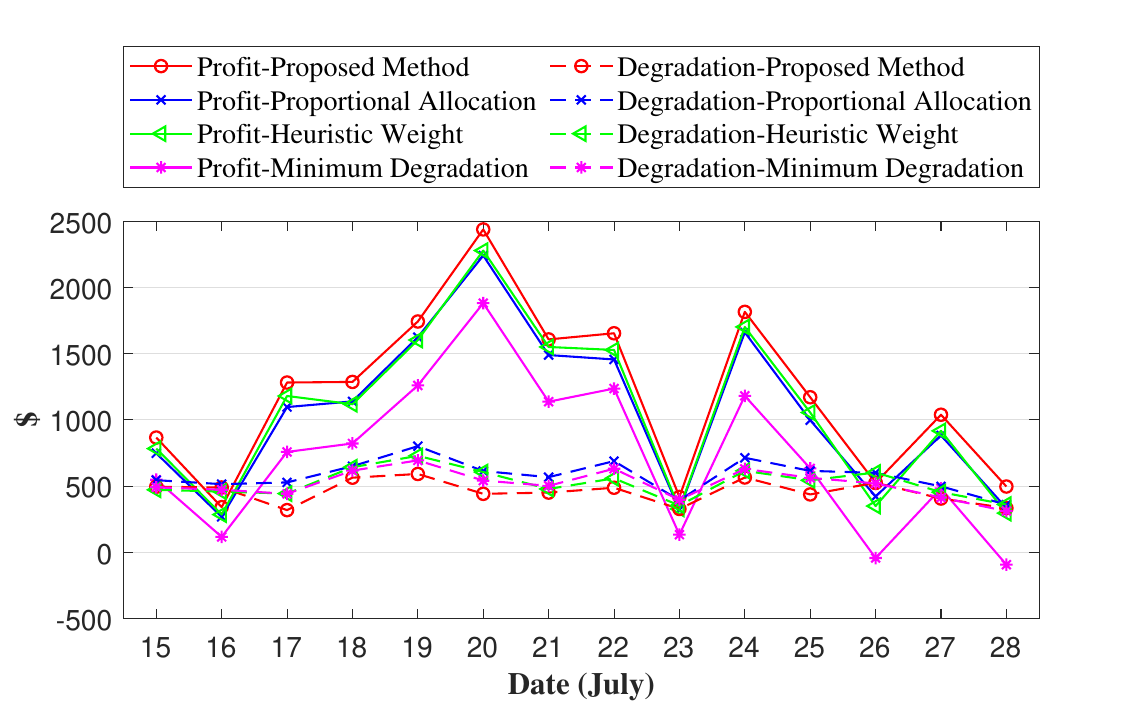}
  \caption{EVA profits and battery degradation costs of July 15-28.}
  \label{fig_wrt_day}
\end{figure}

\textbf{Profits across the Days.} Fig.~\ref{fig_wrt_day} shows the EVA profits and battery degradation costs using different methods, with the prices of July 15-28 and the corresponding RegD signals of July 17-30, respectively. Note that the profits can be negative because the EVA needs to purchase electricity to charge the EVs, and the payments for charging are regarded as constants and therefore not displayed. It can be seen from Fig.~\ref{fig_wrt_day} that due to fluctuations in market prices, the profits of the EVA on different dates fluctuate greatly, and the proposed method can always exceed other methods in profit. Compared to the Proportional Allocation method (blue), the increase ratio of profits using the proposed method (red) ranges between 7.3\% and 48.5\%, and the total profit of the 14 days is increased by 13.6\%.

\subsection{Sensitivity Analysis}
The degradation factor of EVs, the period for updating bids and Lagrangian multipliers, and the granularity of regulation signal scenarios may affect the EVA profit. We investigated the impact of these parameters on the performance of the proposed method.

\begin{figure}[!t]
  \centering
  \includegraphics[width=3.5in]{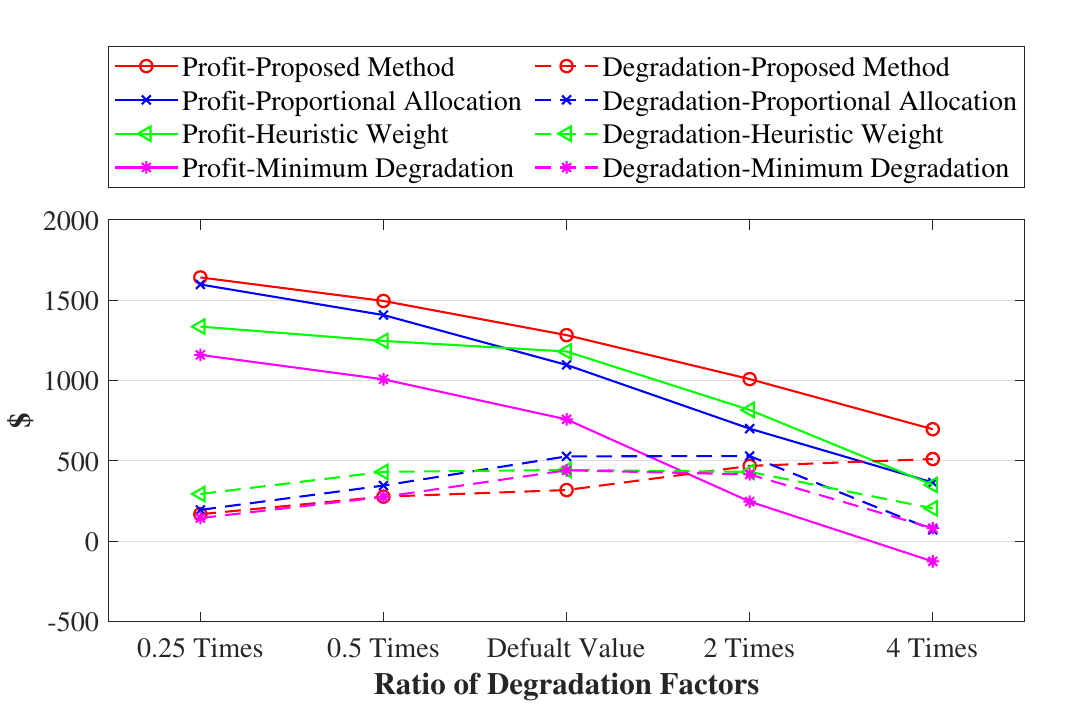}
  \caption{EVA profits and battery degradation costs with different degradation factors.}
  \label{fig_wrt_deg}
\end{figure}

\textbf{Impact of Degradation Factor.} Fig.~\ref{fig_wrt_deg} shows the EVA profits and degradation costs when the degradation factors of EV batteries are set to 0.25-4 times the default value, respectively.
As the degradation factors increase, the EVA profit decreases using whichever method, and the profit using the proposed method always exceeds other methods.
For example, compared with the Proportional Allocation method (blue), when the degradation factors are 1, 2, and 4 times the default value, the EVA profits using the proposed method (red) are 16.7\%, 44.0\%, and 90.8\% higher, respectively.
It should be noted that the EVA aims to maximize its profit, and with certain parameter values, the proposed method can lead to higher degradation costs, i.e., more compensation for EV owners, which is not inconsistent with the goal of the EVA.

\begin{table}[!t]
  \caption{Impact of Different Update Periods}
  \label{tab_profit_period}
  \centering
  \begin{tabular}{cccc}
    \toprule
    \begin{tabular}[c]{@{}c@{}}Update   Period\\      (min) \\ {} \end{tabular} &
    \begin{tabular}[c]{@{}c@{}}Income from\\      the Market\\      (\$)\end{tabular} &
    \begin{tabular}[c]{@{}c@{}}Degradation\\      Cost\\      (\$)\end{tabular} &
    \begin{tabular}[c]{@{}c@{}}EVA \\      Profit\\      (\$)\end{tabular}                     \\ \midrule
    60                         & 1584 & 311 & 1273 \\
    30                         & 1606 & 319 & 1287 \\
    15                         & 1609 & 330 & 1279 \\
    5                          & 1608 & 336 & 1272 \\ \bottomrule
  \end{tabular}
\end{table}

\textbf{Impact of Update Period.} Table~\ref{tab_profit_period} shows the EVA profits and degradation costs using the proposed method with different update periods for bids and Lagrangian multipliers. Theorem~\ref{theorem_invariance} indicates that there is no need to update the Lagrangian multiplier within a time interval (60 min), which is based on the premise that the occurrence frequency of regulation signals is consistent with the forecast and the premise does not hold. Intuitively, by appropriately shortening the update period (e.g., from 60 min to 30 min), the latest information can be better utilized and the impact of uncertainty on profits can be reduced, which is verified by the results in Table~\ref{tab_profit_period}. However, Table also shows that with the update period further shortened, the EVA profit has not increased significantly, or even decreased. This may be caused by the drop in the forecast accuracy of the regulation signal in a shorter period.

\textbf{Impact of Regulation Signal Scenario Granularity.}  As mentioned in Section~\ref{sec_framework_overview}, adopting a smaller scenario granularity in the bidding model means that the discretized regulation signal scenarios will differ less from the actual continuous signals, but the number of scenarios will also increase, resulting in a larger scale and longer solution time for the bidding problem.
Table~\ref{tab_profit_granularity} shows the EVA profits and solution times using the proposed method with different scenario granularities of regulation signals.
Reducing the scenario granularity from 0.2 to 0.05 (adopting [0, 0.2),...,[0.8, 1) and [0, 0.05),...,[0.95, 1) as the discretized regulation signal scenarios, respectively), the number of scenarios increases by 4 times and the EVA profits increased slightly. 
However, further reducing the scenario granularity did not increase profits. This indicates that Assumption~\ref{assumption_express} is reasonable with a fine scenario granularity.
With smaller scenario granularity, the solution time (average time over 100 times of solving at 18:00) of the bidding problem was longer, while the solution time of the power allocation problem remains no more than 0.20 s.
The results verify that the proposed framework is favorable for bidding in the real-time market and online power allocation for regulation deployment.

\begin{table}[!t]
  \caption{Impact of Different Regulation Signal Scenario Granularities}
  \label{tab_profit_granularity}
  \centering
  \begin{tabular}{cccc}
    \toprule
    \begin{tabular}[c]{@{}c@{}}Scenario   \\      Granularity \\ {} \end{tabular} &
    \begin{tabular}[c]{@{}c@{}}EVA \\      Profit\\      (\$)\end{tabular} &
    \begin{tabular}[c]{@{}c@{}}Solution Time - \\      Bidding Problem\\      (s)\end{tabular} &
    \begin{tabular}[c]{@{}c@{}}Solution Time - \\      Power Allocation \\      Problem (s)\end{tabular}                       \\ \midrule
    0.2                        & 1284 & 1.63  & 0.15 \\
    0.1                        & 1287 & 3.52  & 0.17 \\
    0.05                       & 1293 & 8.80  & 0.14 \\
    0.025                      & 1288 & 28.10 & 0.14 \\
    0.01                       & 1287 & 81.21 & 0.16 \\ \bottomrule
  \end{tabular}
\end{table}

\subsection{Discussion}
The above case study verifies the effect of jointly optimizing the bidding and power allocation of an EVA on improving its profits. The advantages of the proposed method derive essentially from the reasonable consideration of the coupling of bidding and power allocation. Here we discuss some specific superiorities of the proposed method over other approaches:

First, it exploits the heterogeneity among different EVs. These differences include not only different arrival and departure time and the electricity to be charged by the EVs, but also the different parameters such as (dis)charging efficiency and degradation cost of EV batteries. Intuitively, the proposed method can give priority to calling resources with low cost. The existing methods, especially the commonly used proportional allocation methods, fail to take advantage of the heterogeneity.

Second, it achieves the correspondence between the bidding model and the power allocation model. In other words, the consideration of power allocation in bidding is consistent with the actual power allocation method in regulation deployment. In the minimum degradation and heuristic weight methods, proportional allocation is used as an assumption in the bidding problem, while the actual power allocation is based on degradation cost/heuristic weights. The mismatch leads to logical inconsistency and loss of EVA profits.

Third, it balances the short-term and long-term profits. The minimum degradation optimizes the short-term profits without considering long-term impacts and is a greedy algorithm that deviates from the optimal. In the proposed method, the impact of power allocation on the future profit is reflected through Lagrange multipliers, so that theoretical optimality can be realized.

\section{Conclusion}\label{sec_conclusion}
An EVA's bidding in the regulation market is coupled with its power allocation strategy for regulation deployment.
Considering this coupling improves the EVA profit.
In this paper, we propose a framework for an EVA to co-optimize its bidding and power allocation in the regulation market.
Power allocation in different regulation signal scenarios is embedded in the bidding model; thus, the costs for providing regulation are simultaneously optimized.
During regulation deployment, the small-scale power allocation model implemented online determines the EV (dis)charging power, where the Lagrange multipliers of the bidding problem enable consideration of the impact on future EVA profits and ensure that the allocation results are optimal for the profit of EVA bidding.
The simulation results verify that co-optimizing bidding and power allocation enables both utilizing the heterogeneous characteristics of the EVs and balancing the short/long-term costs, thus improving the profit of the EVA while reducing the degradation costs.

Although this paper focuses on aggregating EVs, the proposed framework can be generalized for the participation of energy storage systems~\cite{liu_optimal_2018}, heating, ventilation, and air-conditioning (HVAC) systems~\cite{liu_optimal_2021}, etc. in the regulation market.
In these scenarios, the heterogeneous characteristics of different resources need to be considered, and short-term and long-term costs need to be reconciled.

The profit distribution between the EVA and EVs is also coupled with bidding and power allocation. This coupling, along with the coupling between the transportation and the electricity network is for future research.

\appendices
\section{Proof of Theorem\ref{theorem_optimality}}\label{theorem_optimality_app}

\begin{proof}
  First, consider the slackness of constraints (\ref{cons_energyLimit2}) and (\ref{cons_energyLimit3}).
  If (\ref{cons_energyLimit2}) is tight, since $\Delta \hat{t}$ is very small, this indicates that $E^{\rm cur}_{i}$ is at the minimum or maximum level at $\hat{t}^{\rm cur}$.
  Substituting (\ref{cons_energyInit}) into (\ref{cons_energyChange_mid}), we have $E_{t^{\rm cur}, i} = E^{\rm min}_{t^{\rm cur}, i}$ or $E_{t^{\rm cur}, i} = E^{\rm max}_{t^{\rm cur}, i}$, meaning that the bidding problem is not strictly feasible for the constraint of (\ref{cons_energyLimit}), which contradicts Assumption~\ref{assumption_feasible}.

  If (\ref{cons_energyLimit3}) is tight, EV $i$ must be charged at maximum power until it leaves, which is the same for both problems, and therefore, the power allocated to EV $i$ corresponds to the optimal solutions to the two problems.
  Below, we only consider the case where constraints (\ref{cons_energyLimit2}) and (\ref{cons_energyLimit3}) are slack.

  Following convention, we rewrite the bidding problem in the minimizing form (i.e., min. $- Profit(\hat{t}^{\rm cur})$). The parts regarding $t^{\rm cur}$ and $\hat{s}$ in the KKT conditions of the bidding problem are:

  \textbf{a.} Primal constraints: (\ref{cons_balance}) and (\ref{cons_powerLimit}).

  \textbf{b.} Dual constraints: ($\forall i \in I, k \in K$)
  \begin{equation}\label{dual_cons1}
    \underline{\mu^{\rm dis(ch)}_{t^{\rm cur}, \hat{s}, i, k}} \ge 0, \overline{\mu^{\rm dis(ch)}_{t^{\rm cur}, \hat{s}, i, k}} \ge 0
  \end{equation}
  \textbf{\quad c.} Complementary slackness: ($\forall i \in I, k \in K$)
  \begin{flalign}\label{dual_cons2_1}
    & \ \underline{\mu^{\rm dis(ch)}_{t^{\rm cur}, \hat{s}, i, k}} \cdot (- P^{\rm dis(ch)}_{t^{\rm cur}, \hat{s}, i, k}) = 0,  
    &\\ \label{dual_cons2_2}
    & \ \overline{\mu^{\rm dis(ch)}_{t^{\rm cur}, \hat{s}, i, k}} \cdot (P^{\rm dis(ch)}_{t^{\rm cur}, \hat{s}, i, k}
      - u_{t^{\rm cur}, i} \cdot \overline{P^{\rm dis(ch)}_{i, k}}) = 0, 
    &
  \end{flalign}
  \textbf{\quad d.} $\partial L^{\rm B} / \partial x_{t^{\rm cur}, \hat{s}, i, k} = 0$, $\forall i \in I, k \in K$, where $L^{\rm B}$ is the Lagrange multiplier of the bidding problem, $x_{t^{\rm cur}, \hat{s}, i, k} = (P^{\rm dis}_{t^{\rm cur}, \hat{s}, i, k}, P^{\rm ch}_{t^{\rm cur}, \hat{s}, i, k})$ denotes the variables regarding $t^{\rm cur}$ and $\hat{s}$, and $\partial / \partial$ denotes the partial derivative. In $\partial L^{\rm B} / \partial x_{t^{\rm cur}, \hat{s}, i, k}$:
  The corresponding part of the objective function is $  \pi_{t^{\rm cur}, \hat{s}}  Pr^{\rm deg}_{i, k}  (1, 0)^\top \Delta t^{\rm cur}$ (note that $\underset{s \in S}{\Sigma}\pi_s = 1$).
  The constraints (\ref{cons_reg1}), (\ref{cons_reg2}), (\ref{cons_energyInit}) and (\ref{cons_energyLimit}) do not appear.

  The partial of constraint (\ref{cons_balance}) is $\lambda^{\rm bal}_{t^{\rm cur}, \hat{s}} (- 1, 1)^\top$.
  The partial of constraint (\ref{cons_powerLimit}) is $(- \underline{\mu^{\rm dis}_{t^{\rm cur}, \hat{s}, i, k}}, - \underline{\mu^{\rm ch}_{t^{\rm cur}, \hat{s}, i, k}})^\top$,
  $(\overline{\mu^{\rm dis}_{t^{\rm cur}, \hat{s}, i, k}}, \overline{\mu^{\rm ch}_{t^ {\rm cur}, \hat{s}, i, k}})^\top$.
  The partial of constraint 
(\ref{cons_energyChange_mid}), i.e., the modified constraint (\ref{cons_energyChange}), is
  $ \pi_{t^{\rm cur}, \hat{s}} (- \lambda^{\rm E}_{t^{\rm cur}, i} / \eta^{\rm dis}_i, \lambda^{\rm E}_{t^{\rm cur}, i} \cdot \eta^{\rm ch}_i)^\top \Delta t^{\rm cur}$. Therefore, the condition $\partial L^{\rm B} / \partial x_{t^{\rm cur}, \hat{s}, i, k} = 0$ can be expressed as: ($\forall i \in I, k \in K$)
  \begin{equation}\label{dual_cons3}
    \begin{aligned}
       &   \pi_{t^{\rm cur}, \hat{s}}  Pr^{\rm deg}_{i, k} (1, 0)^\top \Delta t^{\rm cur} +
      \lambda^{\rm bal}_{t^{\rm cur}, \hat{s}}  (- 1, 1)^\top                                                                                                              \\
       & - (\underline{\mu^{\rm dis}_{t^{\rm cur}, \hat{s}, i, k}}, \underline{\mu^{\rm ch}_{t^{\rm cur}, \hat{s}, i, k}})^\top
      + (\overline{\mu^{\rm dis}_{t^{\rm cur}, \hat{s}, i, k}}, \overline{\mu^{\rm ch}_{t^{\rm cur}, \hat{s}, i, k}})^\top                                                      \\
       & +  \pi_{t^{\rm cur}, \hat{s}}  (- \lambda^{\rm E}_{t^{\rm cur}, i} / \eta^{\rm dis}_i, \lambda^{\rm E}_{t^{\rm cur}, i} \cdot \eta^{\rm ch}_i)^\top \Delta t^{\rm cur} 
      = 0,                                        
       & 
    \end{aligned}
  \end{equation}
  \quad The KKT conditions for the power allocation problem are:
  \textbf{a.} Primal constraints: (\ref{cons_balance2}) and (\ref{cons_powerLimit2}). According to Assumption~\ref{assumption_forecast}, the deviated power is treated as $0$, and the corresponding constraint of (\ref{cons_powerLimit_unbal}) is not considered. The constraints of (\ref{cons_energyLimit2}) and (\ref{cons_energyLimit3}) are slack and not considered, as mentioned before.

  \textbf{b.} Dual constraints: ($\forall i \in I, k \in K$)
  \begin{equation}
    \underline{\mu^{\rm dis(ch)}_{\hat{t}^{\rm cur}, \hat{s}, i, k}} \ge 0, \overline{\mu^{\rm dis(ch)}_{\hat{t}^{\rm cur}, \hat{s}, i, k}} \ge 0
  \end{equation}
  \textbf{\quad c.} Complementary slackness: ($\forall i \in I, k \in K$)
  \begin{flalign}\label{dual_cons4_1}
    & \ \underline{\mu^{\rm dis(ch)}_{\hat{t}^{\rm cur}, \hat{s}, i, k}} \cdot (- P^{\rm dis(ch)}_{\hat{t}^{\rm cur}, \hat{s}, i, k}) = 0
  & \\ \label{dual_cons4_2}
  & \    \overline{\mu^{\rm dis(ch)}_{\hat{t}^{\rm cur}, \hat{s}, i, k}} \cdot (P^{\rm dis(ch)}_{\hat{t}^{\rm cur}, \hat{s}, i, k}
      - u_{t^{\rm cur}, i} \cdot \overline{P^{\rm dis(ch)}_{i, k}}) = 0
      &
  \end{flalign}
  \textbf{\quad d.} $\partial L^{\rm A} / \partial x_{\hat{t}^{\rm cur}, \hat{s}, i, k} = 0$, $\forall i \in I, k \in K$, where $L^{\rm A}$ is the Lagrange of the power allocation problem and $x_{\hat{t}^{\rm cur}, \hat{s}, i, k} = (P^{\rm dis}_{\hat{t}^{\rm cur}, \hat{s}, i, k}, P^{\rm ch}_{\hat{t}^{\rm cur}, \hat{s}, i, k})$ denotes the variables. In $\partial L^{\rm A} / \partial x_{\hat{t}^{\rm cur}, \hat{s}, i, k}$:
  The corresponding part of the objective function is (substitute (\ref{cons_cost2}) into $Cost_{\hat{t}^{\rm cur}, \hat{s}}$): $ Pr^{\rm deg}_{i, k} (1, 0)^\top \Delta \hat{t} +  (- \lambda^{\rm E}_{t^{\rm cur}, i} / \eta^{\rm dis}_i, \lambda^{\rm E}_{t^{\rm cur}, i} \cdot \eta^{\rm ch}_i)^\top \Delta \hat{t}$.
  The partial of constraint (\ref{cons_balance2}) is $\lambda^{\rm bal}_{\hat{t}^{\rm cur}, \hat{s}} (- 1, 1)^\top$.
  The partial of constraint (\ref{cons_powerLimit2}) is $- (\underline{\mu^{\rm dis}_{\hat{t}^{\rm cur}, \hat{s}, i, k}}, \underline{ \mu^{\rm ch}_{\hat{t}^{\rm cur}, \hat{s}, i, k}})^\top$, $(\overline{\mu^{\rm dis}_{ \hat{t}^{\rm cur}, \hat{s}, i, k}}, \overline{\mu^{\rm ch}_{\hat{t}^{\rm cur}, \hat{s}, i, k }})^\top$. Therefore, the condition $\partial L^{\rm A} / \partial x_{\hat{t}^{\rm cur}, \hat{s}, i, k} = 0$ can be expressed as: ($\forall i \in I, k \in K$)
  \begin{equation}
    \begin{aligned}
       &  Pr^{\rm deg}_{i, k} (1, 0)^\top \Delta \hat{t}
      +   (- \lambda^{\rm E}_{t^{\rm cur}, i} / \eta^{\rm dis}_i, \lambda^{\rm E}_{t^{\rm cur}, i} \cdot \eta^{\rm ch}_i)^\top \Delta \hat{t}\\
       & + \lambda^{\rm bal}_{\hat{t}^{\rm cur}, \hat{s}} (- 1, 1)^\top
      - (\underline{\mu^{\rm dis}_{\hat{t}^{\rm cur}, \hat{s}, i, k}}, \underline{\mu^{\rm ch}_{\hat{t}^{\rm cur}, \hat{s}, i, k}})^\top     \\
       & +  (\overline{\mu^{\rm dis}_{\hat{t}^{\rm cur}, \hat{s}, i, k}}, \overline{\mu^{\rm ch}_{\hat{t}^{\rm cur}, \hat{s}, i, k}})^\top
    \end{aligned}
  \end{equation}
  We now verify that the optimal solution to the bidding problem satisfies the KKT conditions for the power allocation problem. According to Assumption~\ref{assumption_feasible}, the two problems are strictly feasible. Therefore, in the bidding problem, for the optimal power allocated to EVs regarding $t^{\rm cur}$ and $\hat{s}$, i.e., $\{P^{\rm *dis(ch)}_{t^{\rm cur}, \hat{s}, i, k} | i \in I, k \in K\}$,
  $\exists \ \lambda^{\rm *bal}_{t^{\rm cur}, \hat{s}}$,
  $\underline{\mu^{\rm *dis(ch)}_{t^{\rm cur}, \hat{s}, i, k}}$,
  $\overline{\mu^{\rm *dis(ch)}_{t^{\rm cur}, \hat{s}, i, k}}$, and
  $\lambda^{\rm *E}_{t + 1, i}$ that satisfy the KKT conditions (a-d). According to Assumption~\ref{assumption_express}, we have $\hat{s} \in \hat{S}$ and $\pi_{t^{\rm cur}, \hat{s}} > 0$, and therefore $\Delta \hat{t} / (\Delta t^{\rm cur} \pi_{t^{\rm cur}, \hat{s}}) > 0$. $\forall i \in I, k \in K$, let:
  \begin{flalign}\nonumber
    & P^{\rm dis(ch)}_{\hat{t}^{\rm cur}, \hat{s}, i, k} = P^{\rm *dis(ch)}_{t^{\rm cur}, \hat{s}, i, k};
    \ \lambda^{\rm bal}_{\hat{t}^{\rm cur}, \hat{s}} = \Delta \hat{t} / (\Delta t^{\rm cur} \pi_{t^{\rm cur}, \hat{s}})
    \lambda^{\rm *bal}_{t^{\rm cur}, \hat{s}};
    & \\ \nonumber
    & (\underline{\mu^{\rm dis(ch)}_{\hat {t}^{\rm cur}, \hat{s}, i, k}}, \overline{\mu^{\rm dis(ch)}_{\hat{t}^{\rm cur}, \hat{s}, i, k}}) = \frac{\Delta \hat{t}}{\Delta t^{\rm cur} \pi_{t^{\rm cur}, \hat{s}}} (\underline{\mu^{\rm *dis(ch)}_{t^{\rm cur}, \hat{s}, i, k}}, \overline{\mu^{\rm *dis(ch)}_{t^{\rm cur}, \hat{s}, i, k}})
  \end{flalign}
  and substituting $\lambda^{\rm *E}_{t^{\rm cur}, i}$ in, it can be verified that the KKT conditions (a-d) of the power allocation problem are satisfied. Since the two problems are both convex optimizations (linear programming), the optimal solution to the power allocation problem corresponds to the bidding problem.

\end{proof}

\section{Proof of Theorem\ref{theorem_invariance}}\label{theorem_invariance_app}

\begin{proof}
  Denote the optimal solution of ${\rm max.} \ Profit(\hat{t}^{\rm cur})$ as $\{P^{\rm *dis(ch)}_{t, \hat{s}, i, k} | t \in T, \hat{s} \in \hat{S}, i \in I, k \in K\}$, $\{P^{*}_{t}, R^{*}_{t} | t \in T\}$, and $\{E^{*}_{t, i} | t \in T \cup \{t^{\rm end} + 1\}, i \in I\}$.
  Except for the constraints for the initial SOC (\ref{cons_energyInit}) and the SOC change at time interval $t^{\rm cur}$ (\ref{cons_energyChange_mid}), the KKT conditions of ${\rm max.} \ Profit(\hat{t}^{\rm cur})$ naturally satisfy the KKT conditions of ${\rm max.} \ Profit(\hat{t}'^{\rm cur})$ because their parameters and forms are identical. Only KKT conditions related to constraints (\ref{cons_energyInit}) and (\ref{cons_energyChange_mid}) need to be discussed:

  \textbf{a.} Primal constraints. Following regulation signals according to $\{P^{\rm *dis(ch)}_{t^{\rm cur}, \hat{s}, i, k} | i \in I, k \in K \hat{s} \in \hat{S}\}$ from $\hat{t}^{\rm cur}$ to $\hat{t}'^{\rm cur}$, the SOC of EV $i$ at $\hat{t}'^{\rm cur}$ is given by:
  \begin{equation}\label{cons_energyChange_mid2} 
    \begin{aligned}
      E'^{\rm cur}_{i} = & E^{\rm cur}_{i} +  \underset{\hat{s} \in \hat{S}}{\Sigma}\pi_{t^{\rm cur}, \hat{s}} \underset{k \in K}{\Sigma} (P^{\rm *ch}_{t^{\rm cur}, \hat{s}, i, k}\cdot \eta^{\rm ch}_i                                 \\
                         &  - P^{\rm *dis}_{t^{\rm cur}, \hat{s}, i, k} / \eta^{\rm dis}_i)\cdot(\hat{t}'^{\rm cur} - \hat{t}^{\rm cur}),  \ \forall i \in I,
    \end{aligned}
  \end{equation}
We substitute (\ref{cons_energyInit}) and (\ref{cons_energyChange_mid2}) in, and let $P'^{\rm dis(ch)}_{t^{\rm cur}, \hat{s}, i, k} = P^{\rm *dis(ch)}_{t^{\rm cur}, \hat{s}, i, k}$. We have (\ref{cons_energyChange_mid}) of ${\rm max.} \ Profit(\hat{t}'^{\rm cur})$ expressed as (note that $\Delta t^{\rm cur} = \Delta t'^{\rm cur} + (\hat{t}'^{\rm cur} - \hat{t}^{\rm cur})$):
  \begin{equation}\label{cons_energyChange_mid4} 
    \begin{aligned}
       & E^{\rm cur}_{i} + \underset{\hat{s} \in \hat{S}}{\Sigma}\pi_{t^{\rm cur}, \hat{s}} \underset{k \in K}{\Sigma} (P^{\rm *ch}_{t^{\rm cur}, \hat{s}, i, k}\cdot \eta^{\rm ch}_i \\
       & - P^{\rm *dis}_{t^{\rm cur}, \hat{s}, i, k} / \eta^{\rm dis}_i) \cdot \Delta t^{\rm cur} - E_{t + 1, i} = 0,  \ \forall i \in I,
    \end{aligned}
  \end{equation}
  which is consistent with (\ref{cons_energyChange_mid}) of ${\rm max.} \ Profit(\hat{t}^{\rm cur})$ and therefore satisfied.

  \textbf{b.} Conditions related to Lagrange multipliers. For the optimal solution of ${\rm max.} \ Profit(\hat{t}^{\rm cur})$,
  $\exists \ \lambda^{\rm *bal}_{t^{\rm cur}, \hat{s}}$,
  $\underline{\mu^{\rm *dis(ch)}_{t^{\rm cur}, \hat{s}, i, k}}$,
  $\overline{\mu^{\rm *dis(ch)}_{t^{\rm cur}, \hat{s}, i, k}}$, and
  $\lambda^{\rm *E}_{t^{\rm cur}, i}$ that satisfy (\ref{dual_cons1}),(\ref{dual_cons2_1}),(\ref{dual_cons2_2}), and (\ref{dual_cons3}). $\forall i \in I, \hat{s} \in \hat{S}, k \in K$, let:
  \begin{flalign}\nonumber
    & \lambda'^{\rm E}_{t^{\rm cur}, i} = \lambda^{*E}_{t^{\rm cur}, i};
    \ \lambda'^{\rm bal}_{t^{\rm cur}, \hat{s}} = \lambda^{\rm *bal}_{t^{\rm cur}, \hat{s}};
    & \\ \nonumber
       & (\underline{\mu'^{dis(ch)}_{t^{\rm cur}, \hat{s}, i, k}}, \overline{\mu'^{dis(ch)}_{t^{\rm cur}, \hat{s}, i, k}}) = \Delta t'^{\rm cur} / \Delta t^{\rm cur} (\underline{\mu^{\rm *dis(ch)}_{t^{\rm cur}, \hat{s}, i, k}}, \overline{\mu^{\rm *dis(ch)}_{t^{\rm cur}, \hat{s}, i, k}})
  \end{flalign}
  Note that $\Delta t'^{\rm cur} / \Delta t^{\rm cur} > 0$. It is easy to verify that the corresponding KKT conditions of ${\rm max.} \ Profit(\hat{t}'^{\rm cur})$ are satisfied by substituting the above Lagrange multipliers and the optimal solution of ${\rm max.} \ Profit(\hat{t}^{\rm cur})$. Since both problems are convex optimizations, the optimal solution to ${\rm max.} \ Profit(\hat{t}^{\rm cur})$ is also optimal to ${\rm max.} \ Profit(\hat{t}'^{\rm cur})$.
\end{proof}




\ifCLASSOPTIONcaptionsoff
  \newpage
\fi


\bibliographystyle{IEEEtran}
\bibliography{reference}

\begin{IEEEbiography}[{\includegraphics[width=1in,height=1.25in,clip,keepaspectratio]{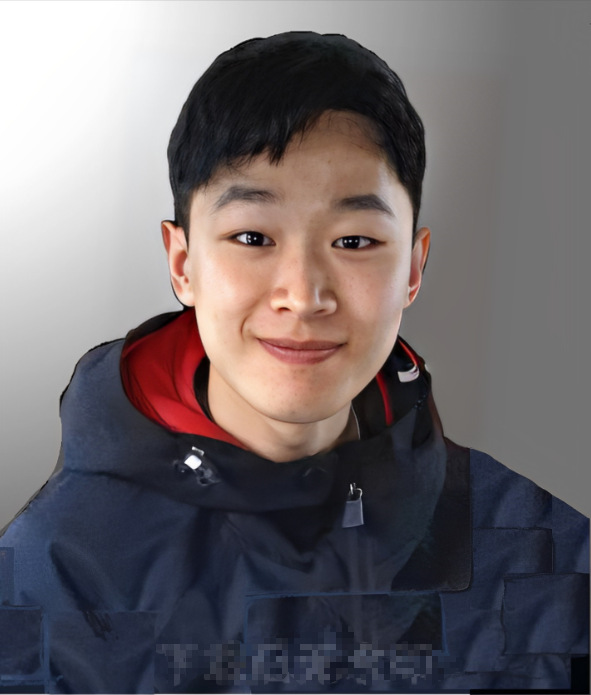}}]{Ruike Lyu}

  received the bachelor's degree in electrical engineering from Tsinghua University, Beijing, China, in 2021. 
  
  He is currently pursuing the Ph.D. degree. His research interests include demand response, electric vehicle, and electricity market.
  
\end{IEEEbiography}

\begin{IEEEbiography}[{\includegraphics[width=1in,height=1.25in,clip,keepaspectratio]{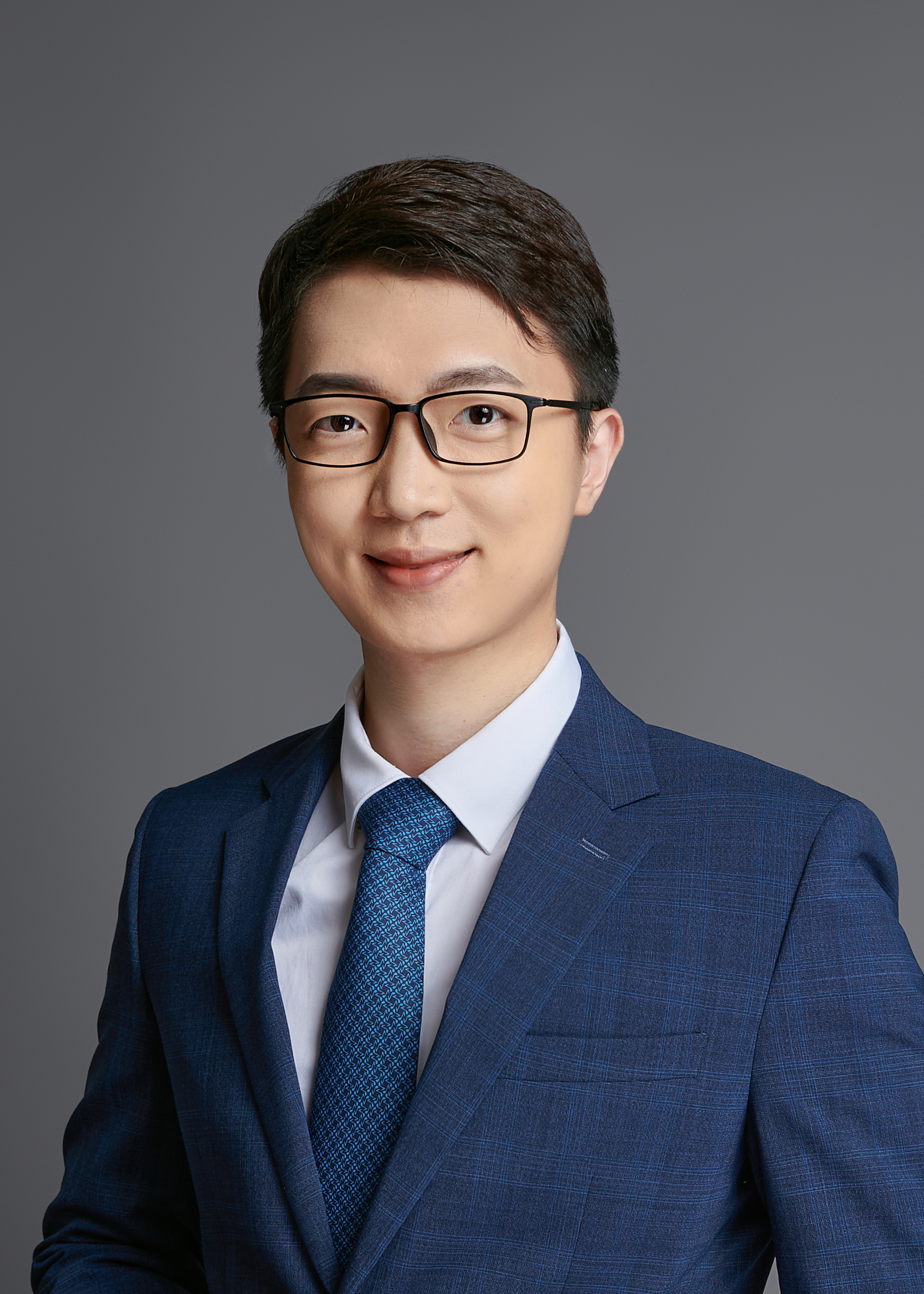}}]{Hongye Guo}

  (S'15-M'20) received the B.S. and Ph.D. degrees in electrical engineering from Tsinghua University, Beijing, China, in 2015 and 2020, respectively. He was a post-doctor in electrical engineering from Tsinghua University, from 2020 to 2022. He was a visiting student researcher with Stanford University, CA, USA, in 2018, and with Illinois Institute of Technology, Chicago, IL, USA, in 2019. 

He is currently a research assistant professor with Tsinghua University. His research interests include electricity markets, demand response, energy economics and machine learning.
  
\end{IEEEbiography}

\begin{IEEEbiography}[{\includegraphics[width=1in,height=1.25in,clip,keepaspectratio]{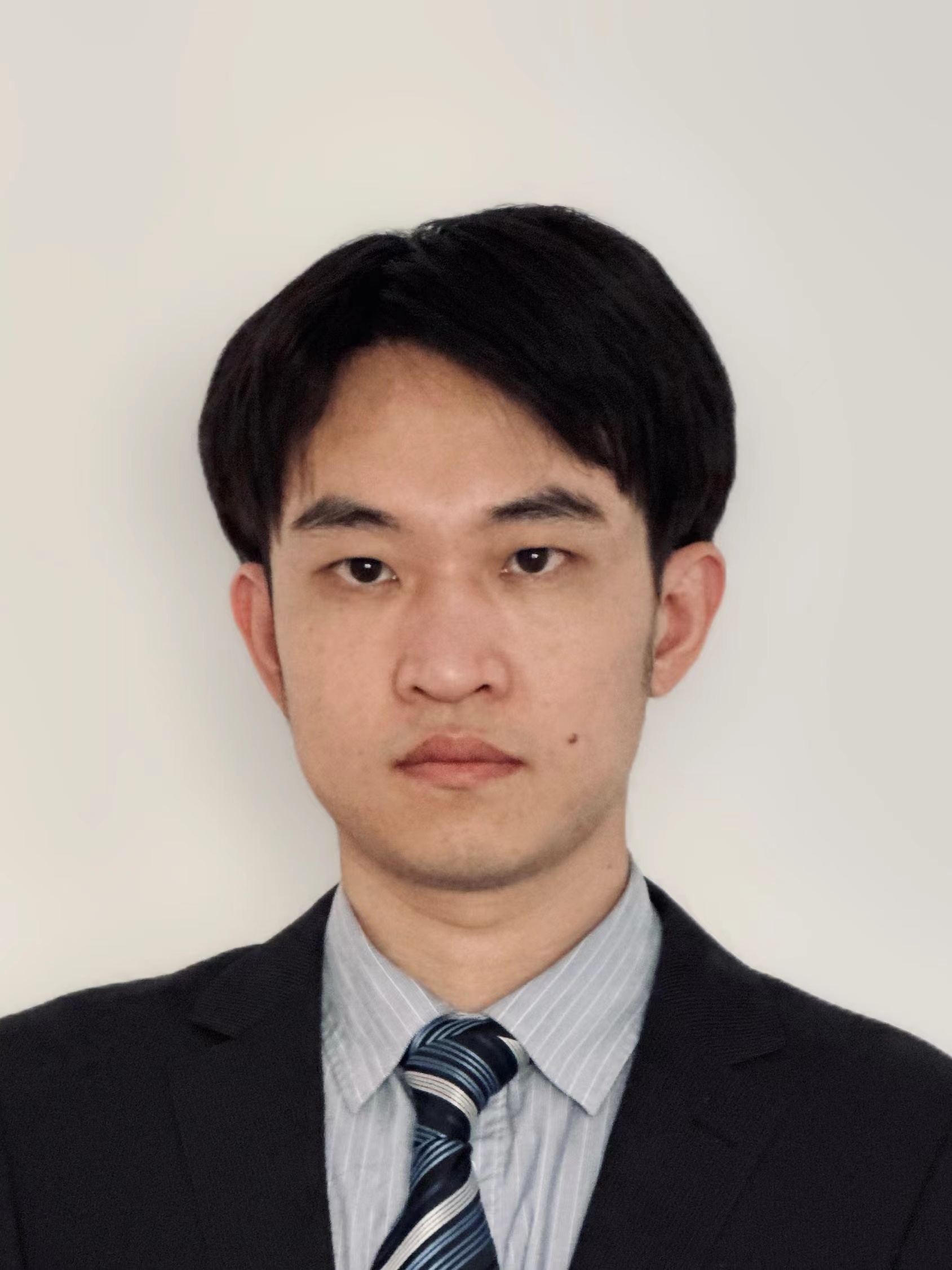}}]{Kedi Zheng}

  (Member, IEEE) received the B.S. and Ph.D. degrees in electrical engineering from Tsinghua University, Beijing, China, in 2017 and 2022, respectively.

He is currently a Postdoctoral Researcher with Tsinghua University. His research interests include data analytics in power systems and electricity markets.
  
\end{IEEEbiography}

\begin{IEEEbiography}[{\includegraphics[width=1in,height=1.25in,clip,keepaspectratio]{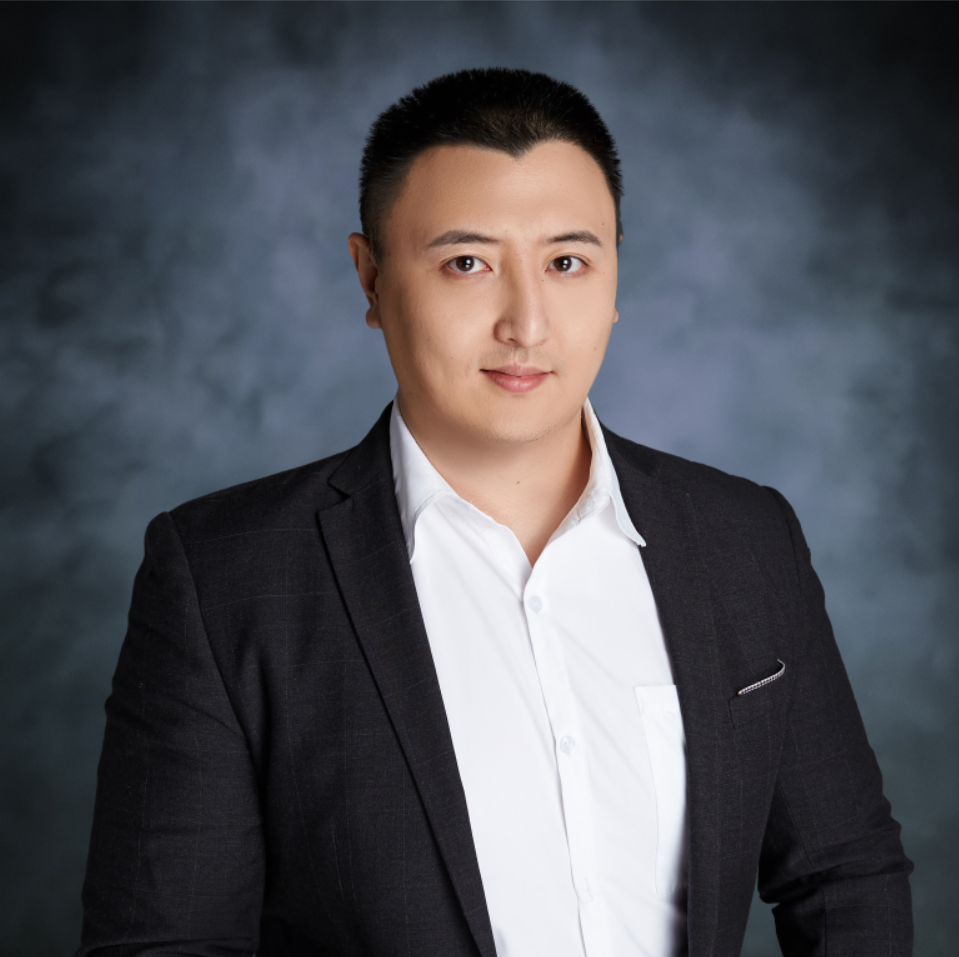}}]{Mingyang Sun}

received the Ph.D. degree from the
Department of Electrical and Electronic Engineering
in Imperial College London, U.K., in 2017. From
2017 to 2019, he was a Research Associate and a
DSI Affiliate Fellow at Imperial College London.
He is now a Professor (Tenure Track) under the
Hundred Talents Program at Zhejiang University. He
is also an Honorary Lecturer at Imperial College
London, UK. His research interests include Artificial
Intelligence in Energy Systems and Cyber Physical
Energy System Security and Control.
  
\end{IEEEbiography}

\begin{IEEEbiography}[{\includegraphics[width=1in,height=1.25in,clip,keepaspectratio]{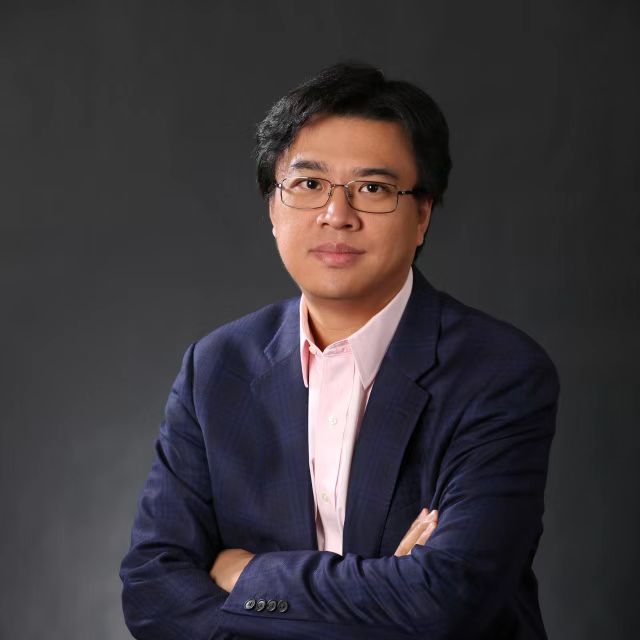}}]{Qixin Chen}

  (Senior Member, IEEE) received the Ph.D. degree from the Department of Electrical Engineering, Tsinghua University, Beijing, China, in 2010. He is currently a Professor with Tsinghua University. 
  
  His research interests include electricity markets, power system economics and optimization, low-carbon electricity, and power generation expansion planning.
  
\end{IEEEbiography}

\end{document}